\documentclass[12pt,a4paper,reqno]{amsart}
\usepackage{latexsym}
\usepackage{tikz}
\usepackage{amsmath,amsfonts,amssymb,epsfig,color}
\usepackage{amsthm}
\usepackage{subfigure}
\usepackage{color}
\definecolor{greencite}{rgb}{0.2,0.6,0.2} 
\definecolor{bluformula}{rgb}{0.1,0.2,0.6}
 \allowdisplaybreaks

\def\R{\mathbb R} 
\def\N{\mathbb N}

\def\H{{\cal H}}

\def\de{\delta}
\def\e{\varepsilon}

\def\vphi{\varphi}
\def\Z{\mathbb Z}

\def\Om{\Omega}

\def\2{\stackrel{2\textrm{-s}}\rightharpoonup}

\def\pa{\partial}

\def\big{\bigskip}

\def\R{\mathbb R}
\def\S{\mathbb S}
\def\H{{\mathcal H}}

\def\vphi{\varphi}

\def\Om{\Omega}
\def\N{\mathbb N}

\def\e{\varepsilon}

\def\pa{\partial}

\def\Div{\mbox{\rm div}\,}

\newcommand{\M}{\mathbb M}

\newcommand{\wto}{\rightharpoonup}
\newcommand{\beq}{\begin{equation}}
\newcommand{\eeq}{\end{equation}}

\newcommand{\medintinrigo}{-\kern -,315cm\int}
\newcommand{\medint}{-\kern -,375cm\int}
\newcommand{\onda}{~}

\newtheorem{theorem}{Theorem}[section]

\newtheorem{lemma}[theorem]{Lemma}
\newtheorem{corollary}[theorem]{Corollary}
\newtheorem{proposition}[theorem]{Proposition}
\newtheorem*{theorem*}{Theorem}
\numberwithin{equation}{section}
\theoremstyle{remark}
\newtheorem{remark}[theorem]{Remark}
\theoremstyle{definition}
\newtheorem{definition}[theorem]{Definition}
\begin{document}

 \title{Reduced models for ferromagnetic thin films with periodic surface roughness}
\author{ M. Morini \& V. Slastikov}
\address[M.\ Morini]{Universit\`a di Parma, Parma, Italy}
\email[Massimiliano Morini]{massimiliano.morini@unipr.it}
\address[V.\ Slastikov]{School of Mathematics, University of Bristol, Bristol, UK}
\email[Valeriy Slastikov]{Valeriy.Slastikov@bristol.ac.uk}


\begin{abstract} 

We investigate the influence of periodic surface roughness in thin ferromagnetic films on shape anisotropy and magnetization behavior inside the ferromagnet. Starting from the full micromagnetic energy and using methods of homogenization and $\Gamma$-convergence we derive a two dimensional local reduced model. Investigation of this model provides an insight on the formation mechanism of {\it perpendicular magnetic anisotropy} and uniaxial anisotropy with an {\it arbitrary} preferred direction of magnetization.

\end{abstract} 

\maketitle

\section{Introduction}

%
%
%
%
%
%

Magnetic anisotropy is one of the fundamental properties of ferromagnetic materials. It is responsible for defining preferred directions of magnetization inside the ferromagnet. The main sources of magnetic anisotropy are magnetocrystalline anysotropy, prescribed by the crystalline structure of the material, and shape anisotropy, induced by the demagnetizing (or stray) field generated by the magnetization distribution inside the ferromagnet. 
In bulk ferromagnets the magnetocrystalline anisotropy provides the leading contribution to magnetic anisotropy and the  demagnetizing field is mainly responsible for formation of multiple domains inside the magnetic sample. On the other hand, in ferromagnetic nanostructures of reduced dimension (thin films, ribbons, nanowires, nanodots) stray field effects may dominate magnetocrystalline anisotropy and become the leading mechanism for choosing preferred magnetization direction.


The geometry of a ferromagnet plays a crucial role in defining the shape anisotropy. It has been observed that in flat ferromagnetic thin films the magnetization vector prefers to be constrained to the plane of the film and align tangentially to the boundary of the film \cite{Aharoni, Gay1986, gioia97, Hubert}. Recent micromagnetic studies of ferromagnetic thin layers, ribbons, and shells with non-trivial curvature of the surface of the film indicate that surface curvature has a significant effect on shape anisotropy, and in ferromagnetic thin structures with non-zero curvature magnetization prefers to be tangent to the surface \cite{carbou01, gaididei17, Gaididei2014, sheka15, Makarov2016_review}. Therefore the dominating effect of the shape anisotropy induced by the stray field is to {\it align magnetization direction tangentially to the surface} of the ferromagnetic nanostructure. This general principle works very well when surface variations happen on a scale larger than thickness of the film (inverse surface curvature is larger than thickness).  However, in the case of rapidly modulated surface, when inverse curvature is of the same order as the thickness of the film, the situation might be different and magnetic anisotropy, dominated by surface curvature effects, may produce preferred directions not tangential to the surface of the film \cite{Bruno1988, Chappert1988, slastikov17}. This behavior might be observed in ultrathin ferromagnetic films with the thickness reaching several monolayers, where the surface roughness can be comparable in amplitude and modulation to the thickness of the film, effectively leading to the large curvature of the film surface.

In this paper we would like to understand the influence of the large surface curvature (or surface roughness) of thin films on the shape anisotropy induced by magnetostatic interaction. We consider the case of periodically modulated thin film surfaces modelling the surface roughness (see Fig.~\ref{fig:oe}). In our study we use the standard continuum model of micromagnetics \cite{Aharoni, Hubert}. In this framework stable magnetization distributions inside a ferromagnet correspond to local minimizers of the micromagnetic energy which after a
suitable nondimensionalization has the following form
\beq \label{smeq1}
{\mathcal E}(M) = d^2 \int_\Omega |\nabla M|^2
+ K \int_\Omega  \phi(M)
+ \int_{{\R}^3} |\nabla u|^2
- 2\int_{\Omega} h_{ext} \cdot M.
\eeq
Here $\Omega \subset \R^3$ is the region occupied by a ferromagnet, $M \colon \Omega \to \S^2$ is the magnetization distribution, and the function $u$ is defined on $\R^3$ and satisfies the following equation
\beq \label{smeq2}
\Div (\nabla u + M \chi (\Omega)) = 0  \hbox{ in } \R^3,
\eeq
\noindent with $\chi (\Omega)$ being the indicator of the set $\Omega$. The applied field is defined by $h_{ext}$, and $\phi$ is the internal anisotropy function. Material parameters $d$ and $K$ correspond to an effective exchange and anisotropy constants, respectively. The four terms of the energy are known as exchange, anisotropy, magnetostatic and Zeeman energies, respectively. Due to the non-convex and non-local nature this variational problem cannot be addressed in its full generality by current analytical methods. 





The standard route to analytically investigate micromagnetic energy \eqref{smeq1} is to consider a range of material and geometric parameters of a ferromagnet where the full three-dimensional model can be reduced to a simpler energy functional, capturing the essence of the magnetization the behavior in ferromagnetic sample \cite{desimone06r}. The derivation and study of the reduced micromagnetic models is by no means a trivial task, but, in general, it is  easier than investigation of the full three-dimensional model. Reduced models have been successfully derived and implemented to explore many magnetic phenomena in ferromagnetic nanostructures, including nanodots \cite{desimone95, slastikov10}, nanowires \cite{harutyunyan16, kuhn07, sanchez09, slastikov12}, thin films \cite{carbou01, desimone06r, desimone02, gioia97, kohn05arma}, and curved structures of reduced dimensions \cite{carbou01, gaididei17, Gaididei2014,  sheka15, slastikov05}. 

\vskip 0.3cm

The main goal of this paper is to obtain a comprehensive reduced model to describe magnetization behavior in ferromagnetic thin films with periodic surface roughness. We concentrate on a regime where the thickness of the film is comparable to the amplitude and the period of thin film surface modulation and derive an effective {\it local} two-dimensional model. This reduced model has been  examined, both analytically and numerically, in the recent paper \cite{slastikov17} and lead to some interesting observations. In particular, it was shown that in the special case of {\it parallel roughness}, when top and bottom surfaces of the layer are parallel, an extreme geometry is responsible for creating a strong uniaxial shape anisotropy with an {\it arbitrary} preferred direction depending on the surface roughness.  This is a rather unexpected outcome suggesting that in certain  regimes a surface roughness in ultrathin ferromagnetic films might lead to a {\it perpendicular magnetic anisotropy} \cite{Chappert1988, 
Johnson1996, Vaz2008}.  In the case of more general roughness, when top and bottom surfaces are different,  several examples have been also  considered where instead the   magnetization prefers to stay in-plane. 



The dimension reduction problems for thin films with periodic surfaces or edges have been extensively studied in the mathematical community in the case where the energy functional has a local energy density, see e.g. \cite{arrieta11, arrieta17, braides00, neukamm10, neukamm13}. The existing results are not directly applicable in our setting due to the nonlocal nature of the stray field energy and the main difficulty in our case comes from homogenizing the magnetostatic contribution.  In order to treat the magnetostatic energy we first identify its leading contribution coming from dipolar interaction of charges at the top and bottom surfaces of thin film. This leading contribution can be represented as an integral with the kernel becoming singular in the limit of vanishing thickness \cite{kohn05arma}. We investigate the homogenized limit of this singular integral and show that the leading order contribution  has a local energy density (similar to the case of flat thin films, see \cite{gioia97}). 

Using methods of $\Gamma$-convergence and two-scale convergence \cite{allaire92, dalmaso} we obtain the limiting behavior of the full micromagnetic energy. Although the treatment of the exchange energy  could be done using the framework of \cite{braides00}, we cannot explicitly use their results due to the more general roughness considered in our paper. Therefore, we adopt the   two-scale convergence approach adapted to dimension reduction problems as developed in \cite{neukamm10} and provide a relatively simple self-contained proof of the $\Gamma$-convergence of the exchange energy. Special care has to be taken due to the fact that magnetization distribution has values on a two dimensional sphere.

\vskip 0.3cm

The paper is organized as follows. In Section~\ref{sec2} we provide a rigorous mathematical formulation of the problem and state our main results  in the Theorem~\ref{th:main}. Section~\ref{sec3} is devoted to the proof of Theorem~\ref{th:main}. We begin our exposition in Section~\ref{subsec:parallel} by finding the limiting behavior of the magnetostatic energy in the case of ``parallel roughness", i.e. when the top and bottom surfaces of the film are exactly the same up to a shift in the vertical direction. The limiting behavior of the magnetostatic energy in the general case  is treated in Section~\ref{subsec:general}. After that, in Section~\ref{sbs:exch} we identify the limiting behavior of the exchange energy. Combining all of the above we arrive at the $\Gamma$-convergence result which completes the proof of the Theorem~\ref{th:main} in Section~\ref{sbs:G}.

\section{Formulation of the problem and statement of the main results}\label{sec2}
In this section we provide a rigorous mathematical set-up of the problem and state out main results in the Theorem~\ref{th:main}. We are interested in proving a $\Gamma$-convergence result and deriving a simplified reduced micromagnetic model (see \eqref{E0gen}). Without loss of generality we are going to consider the case of zero anisotropy and external field, $K=0$ and $h_{ext}=0$ since $\Gamma$-convergence is insensitive to continuous perturbations of the energy functional. 

\vskip 0.2cm

In the following, in order to indicate the generic point $x\in \R^3$ we will use the notation $x=(x', x_3)$, with  $x'=(x_1, x_2)\in \R^2$ and $x_3\in \R$. We also  set $Q:=(0,1) \times (0,1)$ and $\mathbb{S}^2=\pa B(0,1)=\{\xi\in \R^3:\, |\xi|=1\}$.

 Let $f_1, f_2:\R^2\to (0, +\infty)$ be  Lipschitz continuous $Q$-periodic functions, with periodic cell given by  $Q$,  with $f_1<f_2$, and $\omega\subset\R^2$  a bounded open set with Lipschitz boundary. 
 
We will consider three-dimensional  thin film domains with oscillating profiles of the form
\beq\label{dom-gen}
V_\e =\left\{ (x',x_3) \, : \, x' \in \omega \, , \e f_1\left(\frac{x'}{\e}\right)<x_3< \e f_2\left(\frac{x'}{\e}\right) \right\}.
\eeq
\begin{figure}[htbp]
 \begin{center}
 \includegraphics[scale=0.06]{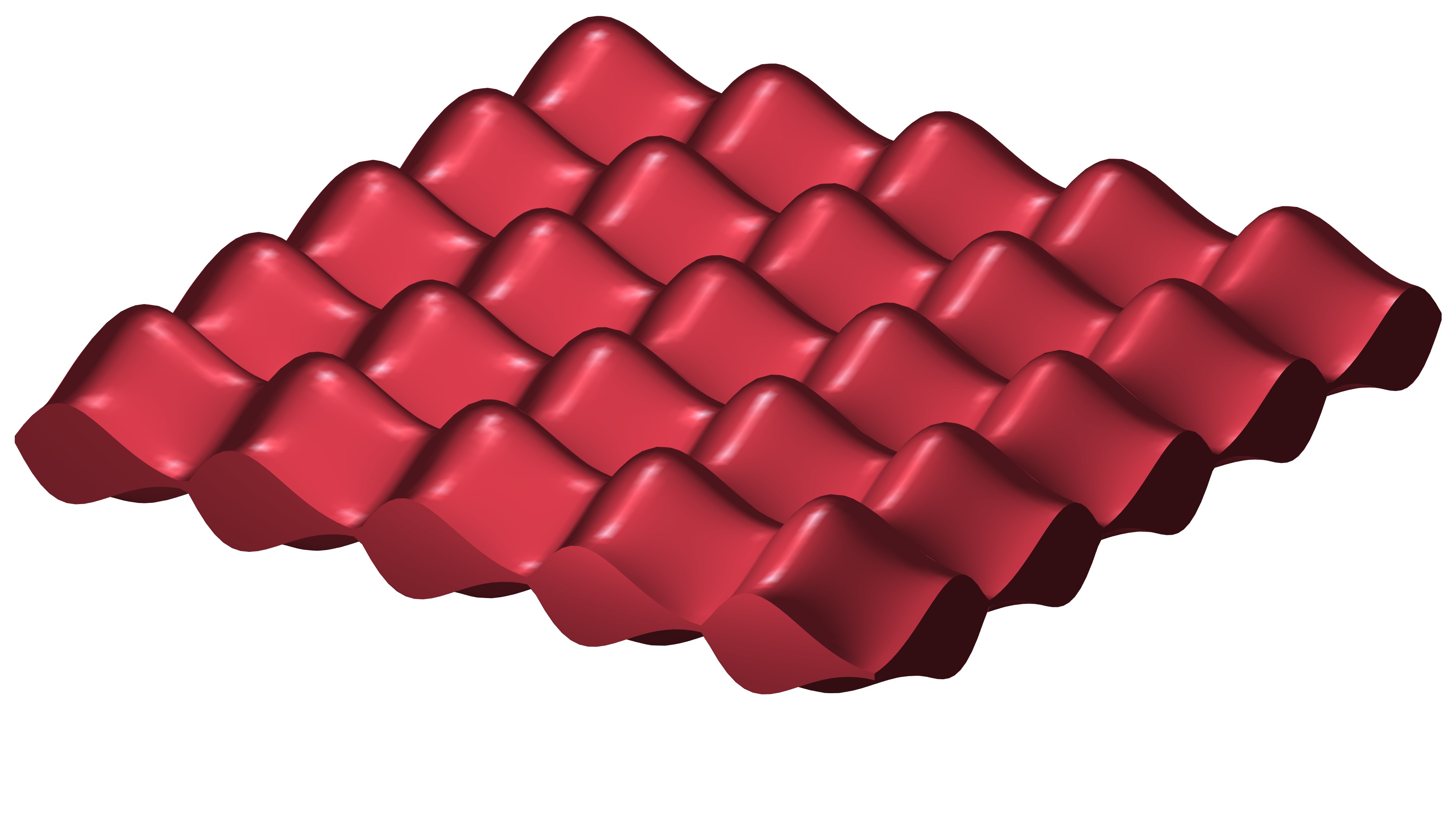}
  \includegraphics[scale=0.06]{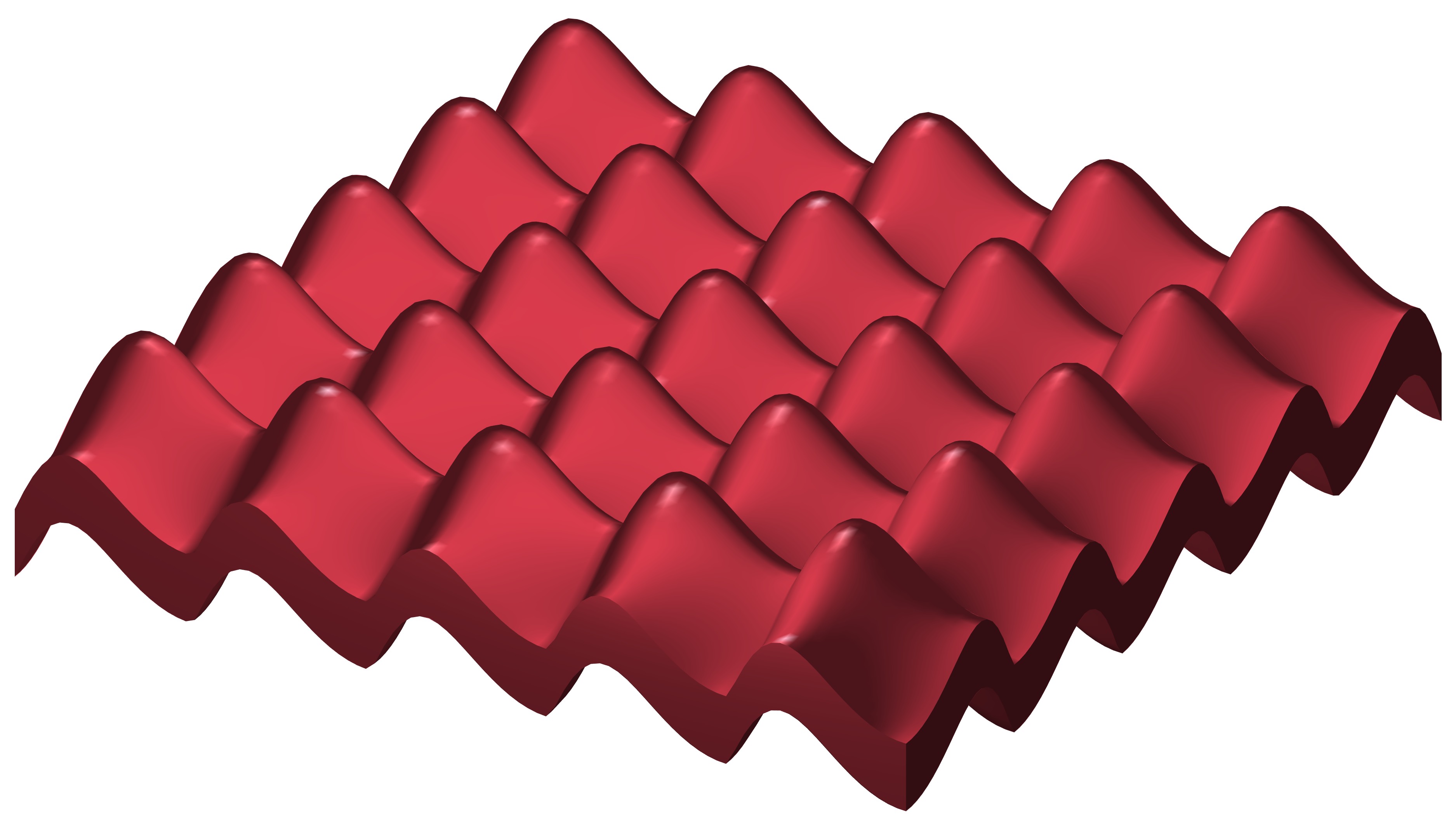}
 \caption{Thin film with generic periodic roughness $V_\e$ (left) and parallel roughness (right) \cite{slastikov17}.}
 \label{fig:oe}
 \end{center}
 \end{figure} 
 
\noindent We recall that given a {\em magnetization} $M\in H^1(V_\e; \mathbb{S}^2)$,  the corresponding micromagnetic energy of the film is  defined as
\beq\label{enMM}
{\mathcal E}_\e (M): = d^2 \int_{V_\e} |\nabla M|^2 +  \int_{\R^3} |\nabla u|^2,
\eeq
where $d>0$ is a material parameter, the so-called {\em exchange constant}, and $u_\e$ is determined as the unique  solution to  
\beq\label{eqe}
\Delta u = {\rm div} ( M  \chi_{V_\e} )  \quad\text{in $\R^3$}
\eeq
in  $\dot{H}^1(\R^3)$, that is, in the homogeneous Sobolev Space obtained as a completion of $C^{\infty}_c(\R^3)$ with respect to the norm $\|u\|_{\dot{H}^1(\R^3)}:=\|\nabla u\|_{L^2(\R^3)}$.
In order to study the limiting behavior of the energy as $\e\to 0^+$, it is convenient to consider the following rescaled energies:
\beq\label{rescEn}
E_\e (m) = {d^2}\int_{\Omega_\e} \left( | \nabla_{x'} m|^2 + \frac{1}{\e^2} |\partial_{x_3} m|^2 \right)\, dx+ \frac{1}{\e}  \int_{\R^3} |\nabla u|^2,
\eeq
defined for all $m\in H^1(\Om_\e; \mathbb{S}^2)$, where 
 $$
\Omega_\e = \left\{ (x', x_3) \, : \, x' \in \omega\, ,  f_1\left(\frac{x'}{\e}\right)<x_3<  f_2\left(\frac{x'}{\e}\right) \right\}
$$
and $u$ now solves  \eqref{eqe} with $M\in H^1(V_\e; \mathbb{S}^2)$ defined by 
$$
M(x',x_3):=m(x',x_3/\e)\,.
$$
Note that 
$$
E_\e (m)=\frac{1}{\e}{\mathcal E}_\e (M)\,.
$$
We also set
$$
Q_{f_1,f_2}:=\{(x', x_3)\in \R^3: x'\in Q\text{ and }f_1(x')<x_3<f_2(x')\}
$$
and denote by $H^1_\#(Q_{f_1,f_2}; \R^3)$ the space of functions $\vphi\in H^1(Q_{f_1,f_2}; \R^3)$ that are $Q$-periodic in the $x'$-variable. 
We will show that the limiting energy is given by the following functional $E_0: H^1(\omega; \mathbb{S}^2)\to [0, +\infty)$ defined by


\beq\label{E0gen}
E_0(m):=d^2\int_{\omega}g_{\mathrm{hom}}(\nabla m)\, dx'+ \int_{\omega}A_{\mathrm{hom}}\, m \cdot m\, dx'
\eeq
for every $m\in H^1(\omega; \mathbb{S}^2)$, where $g_{\mathrm{hom}}:\M^{3\times 2} \to \R$ is  given by
\beq\label{ghom}
g_{\mathrm{hom}}(\xi):=\inf_{\vphi\in H^1_{\#}(Q_{f_1,f_2};\,\R^3)}\int_{Q_{f_1,f_2}}\left(|\xi+\nabla_{y'}\vphi|^2+|\partial_{y_3}\vphi|^2\right)\, dy
\eeq
and {\it constant matrix} $ A_{\mathrm{hom}}$ is defined as
\begin{align}\label{Ahomgen}
 A_{\mathrm{hom}}:&=\frac{1}{4\pi}\int_{Q}\int_{\R^2} \Biggl(  \frac{n_1(x') \otimes n_1(z'+x')}{\sqrt{|z'|^2 + \left| f_1( z'+x') - f_1(x') \right|^2}} 
- \frac{n_1(x') \otimes n_1(z'+x')}{\sqrt{|z'|^2 +1}}   \Biggr)\,  dz'dx' \nonumber \\
 &+\frac{1}{4\pi}\int_{Q}\int_{\R^2}\Biggl(  \frac{n_2(x') \otimes n_2(z'+x') }{\sqrt{|z'|^2 + \left| f_2( z'+x') - f_2(x') \right|^2}} 
- \frac{n_2(x') \otimes n_2(z'+x') }{\sqrt{|z'|^2 +1}}   \Biggr)\, dz'dx' \nonumber\\
&- \frac{1}{2\pi}\int_{Q}\int_{\R^2} \Biggl(  \frac{n_1(x') \otimes n_2(z'+x')}{\sqrt{|z'|^2 + \left| f_2( z'+x') - f_1(x') \right|^2}} 
- \frac{n_1(x') \otimes n_2(z'+x')}{\sqrt{|z'|^2 +1}}   \Biggr)\, dz'dx' \nonumber\\
&+\frac{1}{4\pi}\int_{Q}\int_{\R^2}\Biggl(\frac{I-\mathbf{e_3}\otimes\mathbf{e_3}}{(|z'|^2+1)^{3/2}} -\frac{3(z',0)\otimes (z',0)}{(|z'|^2+1)^{5/2}}\Biggr) \nonumber \\
 & \qquad \qquad \qquad \qquad \qquad  \cdot (f_2(z'+x')-f_1(z'+x'))(f_2(x')-f_1(x'))\, dz'dx'. 
\end{align}
In the above formula we used the notation
\beq\label{nigen}
n_i(x'):=(-\nabla f_i(x'), 1)\quad i=1,2\qquad\text{and}\qquad \mathbf{e_3}:=(0,0,1)\,.
\eeq

We will also show below (see Section~\ref{subsec:parallel}) that in the case of parallel profiles, that is when $f_2= f_1+a$ for a suitable constant $a>0$ (see Figure~\ref{fig:oe}) the expression of  $A_{\mathrm{hom}}$ reduces to the following much simpler formula: 
\begin{align}\label{Ahom}
A_{\mathrm{hom}}=\frac{1}{2\pi}
\int_Q \int_{\R^2} \Biggl[  \frac{ n(x') \otimes n(z'+x')}{\sqrt{|z'|^2 + \left| f( z'+x') - f(x') \right|^2}} - \frac{ n(x') \otimes n(z'+x')}{\sqrt{|z'|^2 + \left|a+ f(z'+x') - f( x') \right|^2}}   \Biggr]\, d z'  d x',
\end{align}
with 
$$
n(x'):=(-\nabla f(x'), 1)\,.
$$
The link between \eqref{rescEn} and \eqref{E0gen} is made precise by the  the following compactness and $\Gamma$-convergence type statement, which represents the main result of the paper. 
\begin{theorem}\label{th:main} The following  statements hold.
\begin{itemize}
\item[i)] (Compactness) Let $\{m_\e\}_\e$ be such that $m_\e\in H^1(\Om_\e; \mathbb{S}^2)$  for every $\e>0$ and 
$$
\sup_{\e} E_\e(m_\e)<+\infty.
$$
Then, there exists $m_0\in H^1(\omega; \mathbb{S}^2)$ and a (not relabelled) subsequence such that 
\beq\label{conve1}
\int_{\Om_\e}|m_\e(x)-m_0(x')|^2\, dx\to 0
\eeq
as $\e\to 0^+$.
\item[ii)] ($\Gamma$-liminf inequality) Let $m_0\in H^1(\omega; \mathbb{S}^2)$  and let $\{m_\e\}_\e$ be such that $m_\e\in H^1(\Om_\e; \mathbb{S}^2)$  for every $\e>0$ and \eqref{conve1} holds. 
 Then
$$
E_0(m_0)\leq \liminf_{\e \to 0} E_\e(m_\e)\,.
$$
\item[iii)] ($\Gamma$-limsup inequality) For any $m_0\in H^1(\omega; \mathbb{S}^2)$, there exists $\{m_\e\}_\e$,   with $m_\e\in H^1(\Om_\e; \mathbb{S}^2)$ for all $\e>0$, such that \eqref{conve1} holds and 
$$
E_0(m_0)= \lim_{\e \to 0} E_\e(m_\e)\,.
$$
\end{itemize}
\end{theorem}
As a consequence of the above theorem, we will  be able to establish the following corollary about the asymptotic behavior of global minimizers. 
\begin{corollary}\label{cor:main}
Let $m_\e\in H^1(\Om_\e; \mathbb{S}^2)$ be a minimizer of $E_\e$.
 Then, up to a (not relabelled) subsequence,  
 $$
\int_{\Om_\e}|m_\e-e_0|^2\, dx\to 0
$$
for a suitable $e_0\in \mathbb{S}^2$ such that 
$$
A_{\mathrm{hom}}\, e_0\cdot e_0=\min_{e\in \mathbb{S}^2} A_{\mathrm{hom}}\, e\cdot e\,.
$$
\end{corollary}

\section{Proofs of the results}\label{sec3}
In this section we collect the proofs of the main results. We treat separately the magnetostatic and the exchange energies.  We start with the study of the magnetostatic energy, which represents the main novelty of the present analysis. In order to simplify the exposition, in Section~\ref{subsec:parallel}  we consider first the case of parallel profiles (see Figure~\ref{fig:oe}). Then, in Section~\ref{subsec:general} we consider the case of general surface roughness, requiring a more intricate analysis, and identify the limiting behavior of the magnetostatic energy in  Proposition~\ref{prop:main}. The $\Gamma$-limit of the exchange energy is investigated in Section~\ref{sbs:exch} (see Propositions~\ref{prop:exch-liminf}, \ref{prop:exch-limsup}). Finally, combining the  aforementioned results we provide the proof of Theorem~\ref{th:main} in Section~\ref{sbs:G}.

\subsection{Study  of the magnetostatic energy:  the case of parallel profiles}\label{subsec:parallel} 
Following   \cite{kohn05arma,slastikov05}, in order to treat the magnetostatic energy we show that its limiting behavior  can be reduced to that of  the energy of magnetic charges at the top and bottom surfaces of the thin layer (see Lemmas~\ref{lm:rappresentazione}-\ref{lm:bulk-boundary}). The core of the analysis is then represented by the study of the latter contribution (see Proposition~\ref{lm:main}) .

In what follows we set $f_1=f$ and $f_2=f+a$, for some $Q$-periodic Lipschitz continuous function $f$ and $a>0$, so that \eqref{dom-gen} becomes 
\beq\label{dom}
V_\e =\left\{ (x',x_3) \, : \, x' \in \omega \, , \e f\left(\frac{x'}{\e}\right)<x_3< a \e+\e f\left(\frac{x'}{\e}\right) \right\}
\eeq
and thus 
$$
\Omega_\e = \left\{ (x', x_3) \, : \, x' \in \omega\, ,  f\left(\frac{x'}{\e}\right)<x_3< a + f\left(\frac{x'}{\e}\right) \right\}\,.
$$
 The typical examples that we might consider is
$$
f(x')= \sin^2(\pi x_1)\sin^2(\pi x_2) \quad \hbox{ or } \quad f(x')= \sin^2(\pi x_1)\,.
$$
We start by recalling the following well-known useful representation formula for the {\em magnetostatic} energy.
\begin{lemma}\label{lm:rappresentazione} Let $u$ solve \eqref{eqe}. 
Then
\begin{align}\label{representation}
4\pi \int_{\R^3}|\nabla u|^2\, dx &= \int_{V_\e}\int_{V_\e}\frac{1}{|x-y|}\Div M(x)\Div M(y)\, dxdy\nonumber\\
&+\int_{\pa V_\e}\int_{\pa V_\e} \frac{1}{|x-y|}(M\cdot \nu_{\e})(x)(M\cdot \nu_{\e})(y)\, d\H^{2}(x)d\H^2(y)\\
&-2\int_{\pa V_\e}\int_{V_\e}\frac{1}{|x-y|}\Div M(y)(M\cdot \nu_\e)(x)\, dy d\H^{2}(x)\,, \nonumber
\end{align}
where $\nu_\e$ denotes the outer unit normal to $V_\e$. 
\end{lemma}
\begin{proof}
See \cite[page 237]{kohn05arma}.
\end{proof}
{\bf Notational warning:} In all the following results (and proofs) $C$ will denote a positive constant possibly depending only on $f$ and $\omega$ (and possibly changing from line to line).

The next lemma provides a simple estimate that will allow us to reduce to the case of $x_3$-independent magnetizations. 
\begin{lemma}\label{lm:media}
Let $M\in H^1(V_\e; \mathbb{S}^2)$.  Set  
$$
\overline M(x'):=\frac{1}{a \e}\int_{\e f(x'/\e)}^{\e a+ \e f(x'/\e)}M(x', x_3)\, dx_3
$$
and let $\bar u$ be the solution to \eqref{eqe} with $M$ replaced by $\overline M$.  Then,
$$
\biggl| \int_{\R^3}|\nabla u|^2\, dx -\int_{\R^3}|\nabla \bar u|^2\, dx\biggr|\leq C \e^{3/2}\biggl\|\frac{\pa M}{\pa x_3}\biggr\|_{L^2(V_\e)}\,.
$$
\end{lemma}
\begin{proof}
The proof can be established arguing as in  \cite[Lemma 3]{kohn05arma}.
\end{proof}
\begin{remark}\label{rm:media}
The previous lemma holds also in the general case \eqref{dom-gen} with the same proof.
\end{remark}
In the next two lemmas we estimate the first and the third terms, respectively, of the representation formula \eqref{representation}.

\begin{lemma}\label{lm:bulk-bulk}
Under the hypothesis and with the notation of the previous lemma we have
$$
\biggl|\int_{V_\e}\int_{V_\e}\frac{1}{|x-y|}\Div \overline M(x')\Div \overline M(y')\, dxdy\biggr|\leq C\e^2\|\Div_{x'}\overline M\|^2_{L^2(\omega)}\,.
$$
\end{lemma}
\begin{proof}
Using the fact that $\overline M$ is independent of $x_3$, one immediately gets
\begin{multline*}
\biggl|\int_{V_\e}\int_{V_\e}\frac{1}{|x-y|}\Div \overline M(x')\Div \overline M(y')\, dxdy\biggr|\\ \leq a^2\e^2
\int_{\omega}\int_{\omega}\frac{1}{|x'-y'|}|\Div_{x'} \overline M(x')|\, |\Div_{y'} \overline M(y')|\, dx'dy'
\leq C\e^2\|\Div_{x'}\overline M\|^2_{L^2(\omega)}\,,
\end{multline*}
where the last estimate follows from the Generalized Young's Inequality (see \cite{lieb-loss}).
\end{proof} 
\begin{lemma}\label{lm:bulk-boundary}
With the notation of the previous lemma we have
$$
\biggl|\int_{\pa V_\e}\int_{V_\e}\frac{1}{|x-y|}\Div \overline M(y')(\overline M\cdot \nu_\e)(x')\, dy d\H^{2}(x)\biggr|
\leq C\e^{3/2}\|\Div_{x'} \overline M \|_{L^2(\omega)}\,.
$$
\end{lemma}
\begin{proof}
Using the inequality 
\beq\label{eq:ineq}
\frac{1}
{\sqrt{|x'-y'|^2+(x_3-y_3)^2}}\leq \frac{1}{\sqrt{|x'-y'|}\sqrt{|x_3-y_3|}}
\eeq
and setting 
\begin{align*}
A:=\biggl|\int_{\omega}\int_{\e f(y'/\e)}^{a\e+ \e f(y'/\e)}&\int_{\omega}\frac{\Div \overline M(y')(\overline M(x')\cdot (-\nabla f(x'/\e), 1))}
{\sqrt{(x'-y')^2+(\e a+ \e f(x'/\e)-y_3)^2}}dx'dy_3dy'
\\
&+\int_{\omega}\int_{\e f(y'/\e)}^{a\e+ \e f(y'/\e)}\int_{\omega}\frac{\Div \overline M(y')(\overline M(x')\cdot (\nabla f(x'/\e), -1))}
{\sqrt{(x'-y')^2+( \e f(x'/\e)-y_3)^2}}dx'dy_3dy'\biggr|\,,
\end{align*}
we have
\begin{align}
&\biggl|\int_{\pa V_\e}\int_{V_\e}\frac{\Div \overline M(y')(\overline M\cdot \nu_\e)(x')}{|x-y|}\, dy d\H^{2}(x)\biggr| \nonumber\\
&\leq \biggl|\int_{\omega}\int_{\pa \omega}\int_{\e f(x'/\e)}^{\e a+ \e f(x'/\e)}\int_{\e f(y'/\e)}^{\e a+ \e f(y'/\e)}\int\frac{\Div \overline M(y')(\overline M\cdot \nu_\e)(x')}{\sqrt{|x'-y'|^2+|x_3-y_3|^2}}\, dy_3 d x_3 d\H^{1}(x')dy'\biggr|+A\nonumber\\
&\leq 
\int_{\omega}\int_{\pa \omega}\frac{|\Div \overline M(y')|}{\sqrt{|x'-y'|}}d\H^{1}(x')dy' \int_{0}^{\e a+ \|f\|_\infty\e }\int_{0}^{\e a+ \|f\|_\infty\e }\frac{1}{\sqrt{|x_3-y_3|}}\, dy_3 d x_3 +A\nonumber\\
&\leq C \e^{3/2} \int_{\omega}\int_{\pa \omega}\frac{|\Div \overline M(y')|}{\sqrt{|x'-y'|}}d\H^{1}(x')dy' 
\leq C \e^{3/2} \|\Div_{x'} \overline M \|_{L^2(\omega)}+A\,.\label{uno}
\end{align}
Since  for $y_3\in (\e f(y'/\e), a\e+ \e f(y'/\e))$ we may find $L>0$  large enough (depending only on $f$ and $a$)   so that 
\begin{multline}\label{keieps}
\biggl|  \frac{1}{\sqrt{(x'-y')^2+( \e f(x'/\e)-y_3)^2}}  -  \frac{1}{\sqrt{(x'-y')^2+(a \e + \e f(x'/\e)-y_3)^2}}\biggr|  \\
 \leq \frac{1}{|x'-y'|} - \frac{1}{\sqrt{(x'-y')^2+ \e^2 L^2}} =:K_\e(x'-y')\,,
\end{multline}
and we can estimate
$$
A \leq C\e \int_\omega \int_\omega  |\Div \overline M(y')| K_\e(x'-y')\,dx' dy'\,.
$$
In turn, by the Generalized Young's inequality and using the fact that
\begin{align}
\int_{\R^2}K_\e(z')\, dz'&=2\pi \int_0^{+\infty}r\left(\frac{1}r-\frac{1}{\sqrt{r^2+\e^2L^2}}\right)\, dr \nonumber \\
&=2\pi \int_0^{+\infty}\frac{\e^2 L^2}{\sqrt{r^2+\e^2L^2}(\sqrt{r^2+\e^2L^2}+r)}\, dr
\leq 2\pi\int_0^\infty \frac{\e^2 L^2}{r^2 + L^2 \e^2} dr =\pi^2\e L \,,\label{keieps2}
\end{align}
we obtain
$$
A \leq C\e  \| K_\e\|_{L^1(\R^2)} \| \Div \overline M \|_{L^2(\omega)}\leq C\e^2 \| \Div \overline M \|_{L^2(\omega)}\,.
$$
Combining the last inequality with \eqref{uno}, we conclude the proof of the lemma.
\end{proof}
The estimates provided by the next two lemmas will be useful in the computing  the limit of the second term in \eqref{representation}. 
\begin{lemma}\label{lm:bd-bd}
With the same notation of the previous lemma, we have
\begin{align*}
\int_{\omega}\int_{\pa \omega}\int_{\e f(x'/\e)}^{a\e+ \e f(x'/\e)}&\biggl|\frac{1}
{\sqrt{|x'-y'|^2+(\e a+ \e f(x'/\e)-y_3)^2}}&\\
&-\frac{1}
{\sqrt{|x'-y'|^2+(\e f(x'/\e)-y_3)^2}}\biggl|dy_3dx'dy'\leq C\e^2
\end{align*}
\end{lemma}
\begin{proof}
We can estimate the integrand as in \eqref{keieps} and \eqref{keieps2} to easily conclude. 
\end{proof}
\begin{lemma}\label{lm:bd-bd2}
We have
$$
\int_{\pa \omega}\int_{\pa \omega}\int_{\e f(x'/\e)}^{a\e+ \e f(x'/\e)}\int_{\e f(y'/\e)}^{a\e+ \e f(y'/\e)}\frac{1}
{\sqrt{|x'-y'|^2+(x_3-y_3)^2}}dy_3dx_3d\H^1(y')d\H^1(x')\leq C \e^{3/2}\,.
$$
\end{lemma}
\begin{proof}
The the proof is  straightforward after recalling \eqref{eq:ineq}.
\end{proof}
We will also need the following simple and rather standard result on the approximation of the identity. It is a particular case of a  more general statement,  however we formulate it only in the form that serves  our purposes.
\begin{lemma}\label{lm:lemmetto}
Let $( K_\e)$ be a family of non-negative   kernels satisfying
\beq\label{eq:lemmetto}
\sup_{\e>0}\int_{\R^2} K_\e(z')\, dz'=:M <+\infty \ \hbox{ and for any fixed } \delta>0 \   \int_{|{ z'}| > \delta} K_\e ({z'})dz' \to 0 \hbox{ as }\e \to 0\,.
\eeq
Let $u_\e\to u$ in $L^1(\R^2; \R^3)$. Then 
$$
\int_{\R^2}\int_{\R^2}K_\e(x'-y')|u_\e(x')-u(y')|\, dx'dy'\to 0
$$ 
as $\e\to 0^+$.
\end{lemma}
\begin{proof}
The proof is rather standard. Observe first that by   \eqref{eq:lemmetto} it easily follows that 
\beq\label{eq:lemmetto2}
w\in C_c(\R^2; \R^3)\Rightarrow \int_{\R^2}\int_{\R^2}K_\e(x'-y')|w(x')-w(y')|\, dx'dy'\to 0\quad\text{ as }\e\to 0^+\,.
\eeq
Fix $\de>0$ and find  $w\in C_c(\R^2; \R^3)$ and $\bar \e>0$  such that $\|w-u\|_{1}\leq \de$ and $\|u_\e-u\|_{1}\leq \de$ 
for all $\e\in (0, \bar \e)$. Then for all such $\e$ we have 
\begin{align*}
\int_{\R^2}&\int_{\R^2}K_\e(x'-y')|u_\e(x')-u(y')|\, dx'dy' \\
&\leq 
 \int_{\R^2}\int_{\R^2}K_\e(x'-y')|u_\e(x')-u(x')|\, dx'dy'+2\int_{\R^2}\int_{\R^2}K_\e(x'-y')|u(x')-w(x')|\, dx'dy'\\
&\quad +\int_{\R^2}\int_{\R^2}K_\e(x'-y')|w(x')-w(y')|\, dx'dy'\\
& = \|K_\e\|_{1}\big( \|u_\e-u\|_{1}+ 2\|w-u\|_{1}\big)+\int_{\R^2}\int_{\R^2}K_\e(x'-y')|w(x')-w(y')|\, dx'dy'\\
& \leq 3 M\de+\int_{\R^2}\int_{\R^2}K_\e(x'-y')|w(x')-w(y')|\, dx'dy'\,,
\end{align*}
where in the last inequality we used the first assumption in \eqref{eq:lemmetto}. Recalling \eqref{eq:lemmetto2} we deduce
$$
\limsup_{\e\to 0}\int_{\R^2}\int_{\R^2}K_\e(x'-y')|u_\e(x')-u(y')|\, dx'dy'\leq 3M\de
$$
and the conclusion follows by the arbitrariness of $\de$.
\end{proof}
The following proposition identifies the limit as $\e\to 0$ of the boundary-boundary term in \eqref{representation} and represents the main brick in the proof of Theorem~\ref{th:main}.   
\begin{proposition}\label{lm:main}
Let $m_0\in L^2(\omega; \mathbb{S}^2)$ and let $(\overline M_\e) \subset L^2(\omega; \R^3)$ be such that $|\overline M_\e|\leq 1$ for all $\e$ and     $\overline M_\e\to m_0$ in $L^2(\omega; \R^3)$. 
Then
$$
\frac{1}{4\pi \e}\int_{\pa V_\e}\int_{\pa V_\e} \frac{1}{|x-y|}(\overline M_\e(x')\cdot \nu_{\e}(x))(\overline M_\e(y')\cdot \nu_{\e}(y))\, d\H^{2}(x)d\H^2(y)
\to  \int_\omega  A_{\mathrm{hom}}\, m_0\cdot m_0\, dx'\,,
$$
where $A_{\mathrm{hom}}$ is the constant matrix defined in \eqref{Ahom}.
\end{proposition}
\begin{proof}
We start by decomposing  $\pa V_\e$ as $\pa V_\e=\Gamma^+_\e\cup \Gamma^-_\e\cup \Gamma^{lat}_\e$, with $
\Gamma^+_\e$ and $\Gamma^-_\e$ denoting the top and the bottom part of $\pa V_\e$, respectively, and $\Gamma^{lat}_\e$ being the lateral boundary.  Observe now that we may split the double integral $\int_{\pa V_\e}\int_{\pa V_\e}$ as
\begin{align}
\int_{\pa V_\e}\int_{\pa V_\e} &= \int_{\Gamma^+_\e}\int_{\Gamma^+_\e}+\int_{\Gamma^-_\e}\int_{\Gamma^-_\e}+2\int_{\Gamma^+_\e}\int_{\Gamma^-_\e}+2\int_{\Gamma^{lat}_\e}\int_{\Gamma^+_\e\cup \Gamma^-_\e}+\int_{\Gamma^{lat}_\e}\int_{\Gamma^{lat}_\e}\nonumber\\
&=2 \int_{\Gamma^+_\e}\int_{\Gamma^+_\e}+2\int_{\Gamma^+_\e}\int_{\Gamma^-_\e}+2\int_{\Gamma^{lat}_\e}\int_{\Gamma^+_\e\cup \Gamma^-_\e}+\int_{\Gamma^{lat}_\e}\int_{\Gamma^{lat}_\e}\,,\label{splitting}
\end{align}
where we used the obvious identity $\int_{\Gamma^+_\e}\int_{\Gamma^+_\e}=\int_{\Gamma^-_\e}\int_{\Gamma^-_\e}$, which follows from the fact that $\Gamma^+_\e$ and $\Gamma^-_\e$ are parallel.
By Lemma~\ref{lm:bd-bd} we easily get 
\beq\label{easily1000}
\int_{\Gamma^{lat}_\e}\int_{\Gamma^+_\e\cup \Gamma^-_\e} \frac{1}{|x-y|}(\overline M_\e(x')\cdot \nu_{\e}(x))(\overline M_\e(y')\cdot \nu_{\e}(y))\, d\H^{2}(x)d\H^2(y)\leq C \e^2\,,
\eeq
while Lemma~\ref{lm:bd-bd2} yields
\beq\label{easily1001}
\int_{\Gamma^{lat}_\e}\int_{\Gamma^{lat}_\e} \frac{1}{|x-y|}(\overline M_\e(x')\cdot \nu_{\e}(x))(\overline M_\e(y')\cdot \nu_{\e}(y))\, d\H^{2}(x)d\H^2(y)\leq C \e^{3/2}\,.
\eeq
 Thus, combining \eqref{splitting}--\eqref{easily1001} we get
\begin{align}\label{semplice}
\lim_{\e \to 0}\frac{1}{4\pi \e}\int_{\pa V_\e}&\int_{\pa V_\e} \frac{1}{|x-y|}(\overline M_\e(x')\cdot \nu_{\e}(x))(\overline M_\e(y')\cdot \nu_{\e}(y))\, d\H^{2}(x)d\H^2(y)\nonumber \\ 
&=  \lim_{\e \to 0} \frac{1}{2\pi \e}\biggl[\int_{\Gamma^+_\e}\int_{\Gamma^+_\e} \frac{1}{|x-y|}(\overline M_\e(x')\cdot \nu_{\e}(x))(\overline M_\e(y')\cdot \nu_{\e}(y))\, d\H^{2}(x)d\H^2(y)\nonumber \\ 
&\qquad\qquad+\int_{\Gamma^+_\e}\int_{\Gamma^-_\e} \frac{1}{|x-y|}(\overline M_\e(x')\cdot \nu_{\e}(x))(\overline M_\e(y')\cdot \nu_{\e}(y))\, d\H^{2}(x)d\H^2(y)\biggr]\nonumber \\
&= \lim_{\e \to 0}   \int_\omega \int_\omega   \Gamma_\e (x',y')\overline M_\e(x')\cdot \overline M_\e(y') \, dx'dy'\,,
\end{align}
where
\begin{multline*}
\Gamma_\e ({ x'},{ y'}) := \\
\frac{1}{2\pi \e}\left( \frac{n(\frac{{ y'}}{\e})\otimes n(\frac{{ x'}}{\e})}{\sqrt{|{x'} -{y'}|^2 + \e^2 \left| f\left(\frac{{x'}}{\e}\right) - f\left(\frac{{ y'}}{\e}\right) \right|^2}}   -\frac{n(\frac{{ y'}}{\e})\otimes n(\frac{{ x'}}{\e})}{\sqrt{|{x'} -{y'}|^2 + \e^2 \left| a+ f\left(\frac{{ x'}}{\e}\right) - f\left(\frac{{ y'}}{\e}\right) \right|^2}}\right),
\end{multline*}
with
$$
n(x'):=\left(-\nabla f(x'), 1\right)\,.
$$
Observe now that there exists $L$ sufficiently large such that
$$
|\Gamma_\e({ x'}, { y'})| \leq  
\frac L{2\pi \e}\left( \frac{1}{{|{ x'} -{' y}| }}  - \frac{1}{\sqrt{|{x'} -{y'}|^2 + \e^2 L^2}}\right)=: \frac L{2\pi \e} K_\e ({x'} - { y'})  
$$
and note that, using also \eqref{keieps2}, we have 
\beq\label{Ktodelta}
\frac L{2\pi \e}\int_{\R^2} K_\e(z')\, dz' \leq \frac{\pi}2 L^2 \quad \hbox{ and for any fixed } \delta>0 \quad  \frac L{2\pi \e} \int_{|{ z'}| > \delta} K_\e ({z'})dz' \to 0 \hbox{ as }\e \to 0\,.
\eeq
We define a $Q$-periodic function 
\begin{align*}
G(x'): = \frac{1}{2\pi}\int_{\R^2}  \Biggl[   &\frac{n(x')\otimes n({ z'}+ x' )}{\sqrt{|{z'}|^2 + \left| f\left(z'+ x'\right) - f\left(x'\right) \right|^2}} -\frac{n(x')\otimes n({ z'}+ x' )}{\sqrt{|{z'}|^2 + \left|a+ f\left(z'+ x'\right) - f\left(x'\right) \right|^2}}\Biggr] \, d{ z'}\,.
\end{align*}
By the  change of variables $z':=(x'-y')/\e$ we obtain
$$
G\left(\frac{{ y'}}{\e} \right) = \int_{\R^2}  \Gamma_\e (x',  y')\, dx'.
$$
Thus,
\begin{align*}
\Biggl|\int_{\omega}  \Gamma_\e ({ x'}, { y'})\overline M_\e({ x'})  \, dx' - G\left(\frac{{ y'}}{\e} \right)m_0({ y'})\Biggr|  &= \Biggl|\int_{\R^2} \Gamma_\e ({ x'}, { y'})(\overline M_\e({ x'})\chi_\omega(x')- m_0 ({ y'})) \, dx'\Biggr| \\
&\leq \frac L{2\pi \e} \int_{\R^2} K_\e ({ x'} - {y'}) |\overline M_\e({x'})\chi_\omega(x') - m_0 ({y'})|   \, dx'
\end{align*}
so that
\begin{multline*}
\int_{\omega}\Biggl|\int_{\omega}  \Gamma_\e ({ x'}, { y'})\overline M_\e({ x'})  \, dx' - G\left(\frac{{ y'}}{\e} \right)m_0({ y'})\Biggr| dy' \\
\leq \frac L{2\pi \e} \int_{\R^2}\int_{\R^2}  K_\e ({ x'} - {y'}) |\overline M_\e({x'})\chi_\omega(x') - m_0 ({y'})\chi_\omega(y')|  \, dx'dy'\to 0\,,
\end{multline*}
where the last limit follows from Lemma~\ref{lm:lemmetto}. 
In turn, using $|\overline M_{\e}(y')| \leq 1$ we have 
\begin{multline*}
 \lim_{\e \to 0}   \int_\omega \int_\omega \Gamma_\e (x',y') \overline M_\e(x')\cdot \overline M_\e(y')  \, dx'dy'\\ =
 \lim_{\e \to 0}\int_\omega  G\left(\frac{{ y'}}{\e} \right)m_0({ y'})\cdot \overline M_\e(y') dy'=
  \int_\omega  A_{\mathrm{hom}}\, m_0\cdot m_0\, dx'\,,
\end{multline*}
where the last equality follows from the Riemann-Lebesgue lemma and the definition of $G$ and $A_{\mathrm{hom}}$.
The conclusion of the lemma follows recalling \eqref{semplice}.
\end{proof}

Combining Lemma~\ref{lm:rappresentazione}, Lemmas~\ref{lm:bulk-bulk}--\ref{lm:bd-bd2} and Proposition~\ref{lm:main} we easily  establish the following asymptotic behavior of the magnetostatic energy.
\begin{proposition}\label{prp:finalparallel} Let $m_0\in H^1(\omega; \mathbb{S}^2)$  and let $\overline M_\e \wto m_0$ weakly in $H^1(\omega; \overline{B(0,1)})$.
 For every $\e>0$ let  $\bar u_\e$ solve   \eqref{eqe} with $M$ replaced by $\overline M_\e$. Then
$$
\frac1\e\int_{\R^3}|\nabla \bar u_\e|^2\, dx\to  \int_\omega  A_{\mathrm{hom}}\, m_0\cdot m_0\, dx'\,,
$$
as $\e\to 0^+$, where $A_{\mathrm{hom}}$ is the matrix defined in \eqref{Ahom}.
\end{proposition}

 \subsection{Study of the magnetostatic energy: the general case}\label{subsec:general} In this section we study the magnetostatic energy  in general domains of the form \eqref{dom-gen}. We note that Lemmas~\ref{lm:media}-\ref{lm:bd-bd2} can be directly transferred to the case of general profiles $f_1$, $f_2$ and therefore we will be referring to them without loss of generality.

\begin{lemma}\label{lm:bd-bdgen}
Let $\overline M'_\e\to m'_0$ in  $L^2 (\omega; \R^2)$, with $|\overline M'_\e|\leq 1$. Then
$$
\frac{1}{\e}\int_{\omega}\int_{\pa \omega}\int_{\e f_1(x'/\e)}^{ \e f_2(x'/\e)}\frac{(\overline M'_\e(x')\cdot \nu_{\omega}(x'))(\overline M'_\e(y')\cdot \nabla f_i(y'/\e))}
{\sqrt{|x'-y'|^2+(x_3- \e f_i(y'/\e))^2}}dx_3d\H^1(x')dy'\to 0
$$
for $i=1, 2$. Here $\nu_\omega$ denotes the outer unit normal to $\pa\omega$.
\end{lemma}
\begin{proof}
Using a change of variable and interchanging integrals, we may rewrite the above integral as
$$
\int_{\pa \omega}\int_{f_1(x'/\e)}^{  f_2(x'/\e)}(\overline M'_\e(x')\cdot \nu_{\omega}(x'))\int_{\omega}\frac{\overline M'_\e(y')\cdot \nabla f_i(y'/\e)}
{\sqrt{|x'-y'|^2+\e^2(x_3- f_i(y' /\e))^2}}dy'dx_3d\H^1(x')
$$
Since for all $x=(x', x_3)$
$$
\frac{\overline M'_\e}
{\sqrt{|x'- \cdot |^2+\e^2(x_3- f_i(\cdot /\e))^2}}\to \frac{m'_0}
{|x'- \cdot |}\qquad\text{in }L^1(\omega; \R^2) 
$$
and $\nabla f_i(\cdot /\e)\stackrel{*}{\rightharpoonup} 0$ weakly-$*$ in  $L^\infty(\omega; \R^2)$ (due to the periodicity of $f_i$), we deduce
that
$$
\int_{\omega}\frac{\overline M'_\e(y')\cdot \nabla f_i(y'/\e)}
{\sqrt{|x'-y'|^2+\e^2(x_3- f_i(y'/\e))^2}}dy'\to 0\qquad\text{for all } x\,. 
$$
Since the above integral is uniformly bounded with respect to $x$, 
the thesis of the lemma follows by the Dominated Convergence Theorem.
\end{proof}
As a consequence of the previous lemma we may now show the following
\begin{lemma}\label{lm:bd-bdgen2}
Let $\overline M_\e=(\overline M'_\e, \overline M^3_\e)\to m_0=(m'_0, m^3_0)$ in  $L^2 (\omega; \R^3)$, with $|\overline M_\e|\leq 1$. Then
\begin{multline*}
\frac{1}{\e}\int_{\omega}\int_{\pa \omega}\int_{\e f_1(x'/\e)}^{ \e f_2(x'/\e)}\frac{(\overline M'_\e(x')\cdot \nu_{\omega}(x'))(\overline M_\e(y')\cdot n_2(y'/\e))}
{\sqrt{|x'-y'|^2+(x_3- \e f_2(y'/\e))^2}}dx_3d\H^1(x')dy'\\
-\frac{1}{\e}\int_{\omega}\int_{\pa \omega}\int_{\e f_1(x'/\e)}^{ \e f_2(x'/\e)}\frac{(\overline M'_\e(x')\cdot \nu_{\omega}(x'))(\overline M_\e(y')\cdot n_1(y'/\e))}
{\sqrt{|x'-y'|^2+(x_3- \e f_1(y'/\e))^2}}dx_3d\H^1(x')dy'\to 0\,.
\end{multline*}
Here $n_1$ and $n_2$ are the vectors defined in \eqref{nigen}.
\end{lemma}
\begin{proof}
Observe that the difference of the two integrals appearing in the statement can be rewritten as 
\begin{align*}
&-\frac{1}{\e}\int_{\omega}\int_{\pa \omega}\int_{\e f_1(x'/\e)}^{ \e f_2(x'/\e)}\frac{(\overline M'_\e(x')\cdot \nu_{\omega}(x'))(\overline M'_\e(y')\cdot \nabla f_2(y'/\e))}
{\sqrt{|x'-y'|^2+(x_3- \e f_2(y'/\e))^2}}dx_3d\H^1(x')dy' \\
&\qquad\qquad\qquad+\frac{1}{\e}\int_{\omega}\int_{\pa \omega}\int_{\e f_1(x'/\e)}^{ \e f_2(x'/\e)}\frac{(\overline M'_\e(x')\cdot \nu_{\omega}(x'))(\overline M'_\e(y')\cdot \nabla f_1(y'/\e))}
{\sqrt{|x'-y'|^2+(x_3- \e f_1(y'/\e))^2}}dx_3d\H^1(x')dy'\\
&+ 
\frac{1}{\e}\int_{\omega}\int_{\pa \omega}\int_{\e f_1(x'/\e)}^{ \e f_2(x'/\e)}\overline M^3_\e (y') (\overline M'_\e(x')\cdot \nu_{\omega}(x'))
\Biggl(\frac{1}{\sqrt{|x'-y'|^2+(x_3- \e f_2(y'/\e))^2}}\\
&\qquad\qquad\qquad\qquad\qquad\qquad\qquad\qquad- \frac{1}{\sqrt{|x'-y'|^2+(x_3- \e f_1(y'/\e))^2}}\Biggl)  dx_3d\H^1(x')dy'\,.
\end{align*}
Now, the first two integrals in the above formula vanish thanks to Lemma~\ref{lm:bd-bdgen}, while the   convergence to zero of the last one can be shown as in  Lemma~\ref{lm:bd-bd}.
\end{proof}
We are ready to prove the main result, which establishes the limiting behavior of the magnetostatic energy. 

\begin{proposition}\label{prop:main}
Let  $\overline M_\e\wto m_0$ weakly in $H^1(\omega; \mathbb{S}^2)$.
Then
$$
\frac{1}{4\pi \e}\int_{\pa V_\e}\int_{\pa V_\e} \frac{1}{|x-y|}(\overline M_\e(x')\cdot \nu_{{\e}}(x))( \overline M_\e(y')\cdot \nu_{{\e}}(y))\, d\H^{2}(x)d\H^2(y)
\to  \int_\omega  A_{\mathrm{hom}} m_0\cdot m_0\, dx'\,,
$$
with $A_{\mathrm{hom}}$ defined in \eqref{Ahomgen}. We recall that $\nu_\e$ stands for the outer unit normal to $\pa V_\e$.
\end{proposition}
\begin{proof}
We start by observing that by Lemmas~\ref{lm:bd-bdgen2} and \ref{lm:bd-bd2} we have
\begin{align}
& \lim_{\e \to 0} \frac{1}{4\pi \e}\int_{\pa V_\e}\int_{\pa V_\e} \frac{1}{|x-y|}(\overline M_\e(x')\cdot \nu_{\e}(x))( \overline M_\e(y')\cdot \nu_{\e}(y))\, d\H^{2}(x)d\H^2(y) \\
&= \lim_{\e \to 0} \Biggl( \frac{1}{4\pi \e}\int_{\omega}\int_{\omega}
\frac{(\overline M_\e(x')\cdot n_1(x'/\e))(\overline M_\e(y')\cdot n_1(y'/\e))}{\sqrt{|x'-y'|^2+\e^2(f_1(x'/\e)-f_1(y'/\e))^2}}dx'dy'\nonumber\\
&+\frac{1}{4\pi \e}\int_{\omega}\int_{\omega}
\frac{(\overline M_\e(x')\cdot n_2(x'/\e))(\overline M_\e(y')\cdot n_2(y'/\e))}{\sqrt{|x'-y'|^2+\e^2(f_2(x'/\e)-f_2(y'/\e))^2}}dx'dy'\nonumber\\
&-\frac{1}{2\pi \e}\int_{\omega}\int_{\omega}
\frac{(\overline M_\e(x')\cdot n_1(x'/\e))(\overline M_\e(y')\cdot n_2(y'/\e))}{\sqrt{|x'-y'|^2+\e^2(f_1(x'/\e)-f_2(y'/\e))^2}}dx'dy' \Biggr)\nonumber\\
& =: \lim_{\e \to 0}I_\e\,. \label{red1}
\end{align}
\vspace{1cm}
Now, notice that 
\begin{align}
&I_\e  = I_\e\nonumber\\
&\pm\frac{1}{4\pi\e} \int_{\omega}\int_{\omega}\frac{(\overline M_\e(x')\cdot n_2(x'/\e)-\overline M_\e(x')\cdot n_1(x'/\e))(\overline M_\e(y')\cdot n_2(y'/\e)-\overline M_\e(y')\cdot n_1(y'/\e))}{\sqrt{|x'-y'|^2+\e^2}}dx'dy' \nonumber\\
&= \frac{1}{4\pi \e}\int_{\omega}\int_{\omega}(\overline M_\e(x')\cdot n_1(x'/\e))(\overline M_\e(y')\cdot n_1(y'/\e))\Biggl(\frac{1}{\sqrt{|x'-y'|^2+\e^2(f_1(x'/\e)-f_1(y'/\e))^2}}\nonumber\\
&\qquad\qquad\qquad\qquad\qquad\qquad\qquad\qquad\qquad\qquad\qquad\qquad\qquad\quad-\frac{1}{\sqrt{|x'-y'|^2+\e^2}}\Biggr)dx'dy'\nonumber \\
&\quad +\frac{1}{4\pi \e}\int_{\omega}\int_{\omega}(\overline M_\e(x')\cdot n_2(x'/\e))(\overline M_\e(y')\cdot n_2(y'/\e))\Biggl(\frac{1}{\sqrt{|x'-y'|^2+\e^2(f_2(x'/\e)-f_2(y'/\e))^2}}\nonumber\\
&\qquad\qquad\qquad\qquad\qquad\qquad\qquad\qquad\qquad\qquad\qquad\qquad\qquad\quad-\frac{1}{\sqrt{|x'-y'|^2+\e^2}}\Biggr)dx'dy'\nonumber \\
&\quad - \frac{1}{2\pi \e}\int_{\omega}\int_{\omega}(\overline M_\e(x')\cdot n_1(x'/\e))(\overline M_\e(y')\cdot n_2(y'/\e))\Biggl(\frac{1}{\sqrt{|x'-y'|^2+\e^2(f_1(x'/\e)-f_2(y'/\e))^2}}\nonumber\\
&\qquad\qquad\qquad\qquad\qquad\qquad\qquad\qquad\qquad\qquad\qquad\qquad\qquad\quad-\frac{1}{\sqrt{|x'-y'|^2+\e^2}}\Biggr)dx'dy'\nonumber \\
&\quad+ \frac{1}{4\pi\e} \int_{\omega}\int_{\omega}\frac{(\overline M'_\e(x')\cdot \nabla(f_2-f_1)(x'/\e)) (\overline M'_\e(y')\cdot \nabla(f_2-f_1)(y'/\e))}{\sqrt{|x'-y'|^2+\e^2}}dx'dy'\nonumber\\
&=:I^1_\e+I^2_\e+I^3_\e+I^4_\e\,.\label{red2}
\end{align}
Here we used again the notation $\overline M_\e=(\overline M'_\e, M^3_\e)$.
The limits of $I^1_\e$, $I^2_\e$, and $I^3_\e$ can be computed arguing exactly as in the proof of Lemma~\ref{lm:main}. We obtain
\beq\label{red3}
I^1_\e\to \int_\omega  A_{\mathrm{hom},1} m_0\cdot m_0\, dx'\,, \
I^2_\e\to \int_\omega  A_{\mathrm{hom},2} m_0\cdot m_0\, dx'\,, \ \text{and}\ 
I^3_\e\to \int_\omega  A_{\mathrm{hom},3} m_0\cdot m_0\, dx'\,,
\eeq
where
\begin{align*}
 A_{\mathrm{hom},1}:=&\frac{1}{4\pi}\int_{Q}\int_{\R^2}\Biggl(  \frac{n_1(x') \otimes n_1(z'+x') }{\sqrt{|z'|^2 + \left| f_1( z'+x') - f_1(x') \right|^2}} - \frac{n_1(x') \otimes n_1(z'+x') }{\sqrt{|z'|^2 +1}}   \Biggr)\, dz'dx' \nonumber \\
A_{\mathrm{hom},2}:= &\frac{1}{4\pi}\int_{Q}\int_{\R^2}\Biggl(  \frac{n_2(x') \otimes n_2(z'+x') }{\sqrt{|z'|^2 + \left| f_2( z'+x') - f_2(x') \right|^2}} - \frac{n_2(x') \otimes n_2(z'+x') }{\sqrt{|z'|^2 +1}}   \Biggr)\, dz'dx' \nonumber\\
A_{\mathrm{hom},3}:= &- \frac{1}{2\pi}\int_{Q}\int_{\R^2}\Biggl(  \frac{n_1(x') \otimes n_2(z'+x') }{\sqrt{|z'|^2 + \left| f_2( z'+x') - f_1(x') \right|^2}} 
- \frac{n_1(x') \otimes n_2(z'+x') }{\sqrt{|z'|^2 +1}}   \Biggr)\, dz'dx'\,.
\end{align*}
We are left with studying the behavior of $I^4_\e$. 
In order to deal with such a term, we set  $g:=f_2-f_1$ and we note that integration by parts yields
\begin{align}
4\pi I^4_\e = &\, \e \int_{\omega}\int_{\omega}\Div_{y'}\biggl[\overline M'_\e(y')\Div_{x'}\biggl(\frac{\overline M'_\e(x')}{\sqrt{|x'-y'|^2+\e^2}}\biggr)\biggr]g(x'/\e)g(y'/\e)dx'dy'\nonumber\\
&+\int_{\pa \omega}(\overline M'_\e\cdot \nu_{\omega})(x')g(x'/\e)\int_{\omega}\frac{\overline M'_\e(y')\cdot \nabla g(y'/\e)}{\sqrt{|x'-y'|^2+\e^2}}\, dy' d\H^1(x')\nonumber\\
&-\e \int_{\pa \omega}\int_{\omega}(\overline M'_\e\cdot \nu_{\omega})(y')\Div_{x'}\biggl(\frac{\overline M'_\e(x')}{\sqrt{|x'-y'|^2+\e^2}}\biggr)g(x'/\e)g(y'/\e)\,dy' d\H^1(x')\nonumber\\
=:&  J^1_\e+J^2_\e+J^3_\e\,. \label{red5}
\end{align}
Arguing exactly as in the proof of Lemma~\ref{lm:bd-bdgen}, the $L^{\infty}$ weak-$*$ convergence to $0$ of $\nabla g(\cdot /\e)$ easily yields that
\beq\label{J2}
J^2_\e\to 0\,.
\eeq
Moreover, for a sufficiently large $C>0$, we have
\begin{align}
|J^3_\e|& \leq \e \|g_{\infty}\|_\infty^2 \left( \int_{\pa \omega}\int_{\omega}\frac{|\Div_{x'}\overline M'_\e|}{\sqrt{|x'-y'|^2+\e^2}}\, dx'd\H^1(y')
+ \int_{\pa \omega}\int_{\omega}\frac{|x'-y'|}{(|x'-y'|^2+\e^2)^{3/2}}\, dx'd\H^1(y') \right) \nonumber \\
&\leq \sqrt{\e} \|g_{\infty}\|_\infty^2 \int_{\pa \omega}\int_{\omega}\frac{|\Div_{x'}\overline M'_\e|}{\sqrt{|x'-y'|}}\, dx'd\H^1(y') +C\e  \int_0^C\frac{r^2}{(r^2+\e^2)^{3/2}}\, dr\nonumber \\
&\leq C\sqrt{\e}+ C\e  \int_0^C\frac{r}{(r^2+\e^2)}\, dr\to 0\,,
\label{J3}
\end{align}
where the last convergence follows by explicit computation of the integral. Note that in the last inequality we have also used the fact that $\Div_{x'}\overline M'_\e$ is bounded in $L^2$. In order to deal with $J^1_\e$, we expand  the double divergence term to get
\begin{align*}
J^1_\e & =  
\e \int_{\omega}\int_{\omega}\frac{\Div_{x'}\overline M'_\e(x')\Div_{y'}\overline M'_\e(y')}{\sqrt{|x'-y'|^2+\e^2}}g(x'/\e)g(y'/\e)dx'dy' \\
&\quad+2\e \int_{\omega}\int_{\omega}\frac{\Div_{x'}\overline M'_\e(x')\overline M'_\e(y')\cdot(x'-y')}{(|x'-y'|^2+\e^2)^{3/2}}g(x'/\e)g(y'/\e)dx'dy'\\
&\quad +\e  \int_{\omega}\int_{\omega}\frac{\overline M'_\e(x')\cdot \overline M'_\e(y')}{(|x'-y'|^2+\e^2)^{3/2}}g(x'/\e)g(y'/\e)dx'dy'\\
&\quad -3\e \int_{\omega}\int_{\omega}\frac{[\overline M'_\e(x')\cdot(x'-y')] [\overline M'_\e(y')\cdot(x'-y')]}{(|x'-y'|^2+\e^2)^{5/2}}g(x'/\e)g(y'/\e)dx'dy' \\
&=: J^{1,1}_\e+J^{1,2}_\e+J^{1,3}_\e+J^{1,4}_\e\,.
\end{align*}
Note that 
$$
J^{1,1}_\e\leq \|g\|_\infty^2  \int_{\omega}\int_{\omega}K_\e(x'-y')|\Div_{x'}\overline M'_\e(x')|\,|\Div_{y'}\overline M'_\e(y')|\, dx'dy'\,,
$$
where we set 
$$
K_\e(z'):=\frac{\e}{\sqrt{|z'|^2+\e^2}}\,.
$$
Using the fact that  $\|K_\e\|_{L^1(B)} \leq C \e$, where $B$ is a sufficiently large ball containing $\omega-\omega$,  and that
$\Div \overline M_\e$ is bounded in $L^2$, we deduce from the Generalized Young's Inequality that
$J^{1,1}_\e\to 0$.
Analogously,
$$
J^{1,2}_\e\leq 2\|g\|_\infty^2  \int_{\omega}\int_{\omega}K'_\e(x'-y')|\Div_{x'}\overline M'_\e(x')|\, dx'dy'\,,
$$
with
$$
K'_\e(z'):=\frac{\e |z'|}{(|z'|^2+\e^2)^{3/2}}\,.
$$
Since  $\|K'_\e\|_{L^1(B)}\to 0$ (see \eqref{J3}), we also have $J^{1,2}_\e\to 0$ using Generalized Young's Inequality. Thus,
$$
\lim_{\e\to 0} J^1_\e=\lim_{\e\to 0}(J^{1,3}_\e+J^{1,4}_\e)\,.
$$
The last limit can be now computed arguing as in  the proof of Lemma~\ref{lm:main} to get
\beq\label{J1}
\lim_{\e\to 0} J^1_\e=\lim_{\e\to 0}(J^{1,3}_\e+J^{1,4}_\e)= 4\pi \int_\omega A'_{\mathrm{hom}, 4} m'_0\cdot m'_0 dx'\,,
\eeq
with
$$
A'_{\mathrm{hom}, 4}:=
\frac{1}{4\pi}\int_{Q}\int_{\R^2}\Biggl(\frac{Id}{(|z'|^2+1)^{3/2}}-\frac{3z \otimes z'}{(|z'|^2+1)^{5/2}}\Biggr)g(z'+x')g(x')\, dz'dx'\,.
$$
We reproduce here the argument for the reader's convenience.  First of all, note that we can write
$$
J^{1,3}_\e+J^{1,4}_\e=\int_{\omega}\int_{\omega}\hat \Gamma_\e(x', y') \overline M'_\e(x') \overline M'_\e(y')\, dx'dy'\,,
$$
where
$$
\hat \Gamma_\e(x', y'):=
\e \Biggl(\frac{Id}{(|x'-y'|^2+\e^2)^{3/2}}-\frac{3(x'-y') \otimes (x'-y')}{(|x'-y'|^2+\e^2)^{5/2}}\Biggr)g(x'/\e)g(y'/\e)\,,
$$
and note that
$$
 |\Gamma_\e(x', y')|\leq
\e\|g\|_\infty^2 \Biggl|\frac{Id}{(|x'-y'|^2+\e^2)^{3/2}}-\frac{3(x'-y') \otimes (x'-y')}{(|x'-y'|^2+\e^2)^{5/2}}\Biggr|=:\hat K_\e(x'-y')\,,
$$
with $\hat K_\e$ satisfying \eqref{Ktodelta} (with $\hat K_\e$ in place of $\frac{L}{2\pi \e} K_\e$). Moreover, a change of variables shows that 
$$
\hat G\left(\frac{{ y'}}{\e} \right) = \int_{\R^2} \hat \Gamma_\e (x',  y')\, dx'\,,
$$
where
$$
\hat G(x'):= \int_{\R^2}\Biggl(\frac{Id}{(|z'|^2+1)^{3/2}}-\frac{3z \otimes z'}{(|z'|^2+1)^{5/2}}\Biggr)g(z'+x')g(x')\, dz'\,.
$$
We can now proceed as in the last part of the proof of Lemma~\ref{lm:main} to show that 
$$
\int_{\omega}\Biggl|\int_{\omega}  \hat \Gamma_\e ({ x'}, { y'}) \overline M'_\e({ x'})  \, dx' - \hat G\left(\frac{{ y'}}{\e} \right)m'_0({ y'})\Biggr| dy'\to 0
$$
and, in turn, 
\begin{multline*}
 \lim_{\e \to 0}   \int_\omega \int_\omega \hat \Gamma_\e (x',y')  \overline M'_\e(x')\cdot  \overline M'_\e(y')  \, dx'dy'\\ =
 \lim_{\e \to 0}\int_\omega  \hat G\left(\frac{{ y'}}{\e} \right)m'_0({ y'})\cdot  \overline M'_\e(y') dy'=4\pi
  \int_\omega  A'_{\mathrm{hom}, 4} m'_0\cdot m'_0\, dx'\,.
\end{multline*}
This establishes \eqref{J1}
Collecting \eqref{red3}--\eqref{J1} we conclude the proof of the proposition.
\end{proof}
As at the end of Section~\ref{subsec:parallel}, we can combine the previous results to obtain the following:
\begin{proposition}\label{prp:finalgeneral} Let $m_0\in H^1(\omega; \mathbb{S}^2)$  and let $\overline M_\e \wto m_0$ weakly in $H^1(\omega; \overline{B(0,1)})$.
 For every $\e>0$ let  $\bar u_\e$ solve   \eqref{eqe} with $M$ replaced by $\overline M_\e$. Then
$$
\frac1\e\int_{\R^3}|\nabla \bar u_\e|^2\, dx\to  \int_\omega  A_{\mathrm{hom}}\, m_0\cdot m_0\, dx'\,,
$$
as $\e\to 0^+$, where $A_{\mathrm{hom}}$ is the matrix defined in \eqref{Ahomgen}.
\end{proposition}

\subsection{Study of the exchange energy}\label{sbs:exch}
In this section we identify the limiting exchange energy. 
We start with the following simple extension argument. 
\begin{lemma}\label{lm:extension}
Let $M>\max\{\|f_1\|_\infty, \|f_2\|_\infty\}$ and set $\Om^M:=\omega\times (0,M)$. Let $\{m_\e\}$  be such that $m_\e\in H^1(\Om_\e; \mathbb{S}^2)$ for every $\e>0$ and  
\beq\label{supe}
\sup_{\e>0} \int_{\Omega_\e} \left( | \nabla_{x'} m_\e|^2 + \frac{1}{\e^2} |\partial_{x_3} m_\e|^2 \right)\, dx<+\infty. 
\eeq
Then for every $\e>0$ there exists $\tilde m_\e\in H^1(Q_M; \S^2)$ such that $\tilde m_\e=m_\e$ in $\Om_\e$ and 
\beq\label{supme}
\sup_\e\int_{\Omega^M} \left( | \nabla_{x'} \tilde m_\e|^2 + \frac{1}{\e^2} |\partial_{x_3} \tilde m_\e|^2 \right)\, dx<+\infty. 
\eeq
\end{lemma}
\begin{proof}
The required extension is obtained through repeated vertical reflections with respect to the graphs of $f_1$ and $f_2$. More precisely, for every $k\in \N$, $k\geq 3$, we set 
$f_k:=f_2+(k-2)(f_2-f_1)$ and for $k\in \Z$, with $k\leq 0$,  set $f_k:=f_1+ (k-1)(f_2-f_1)$.
Moreover, for every $\e>0$ and $k\in \Z$ denote
$$
\Om_\e^k:=\left\{(x', x_3):\, x'\in \omega, f_k\left(\frac{x'}{\e}\right)< x_3<f_{k+1}\left(\frac{x'}{\e}\right)\right\}
$$
In particular, note that $\Om_\e^1=\Om_\e$. Set $m_\e^1:=m_\e$ on $\Om_\e$ and inductively define $m_\e^k$ on $\Om_\e^k$ as
$$
m_\e^k(x', x_3):=\begin{cases}
m_\e^{k-1}\left(x', 2 f_k\left(\frac{x'}{\e}\right)-x_3\right) & \text{if }k\geq 2\,, \vspace{7pt}\\

m_\e^{k+1}\left(x', 2 f_{k+1}\left(\frac{x'}{\e}\right)-x_3\right) & \text{if }k\leq 0\,.
\end{cases}
$$
Finally, we let $\tilde m_\e: \omega\times\R \to \S^2$ be   defined   as $\tilde m_\e:=m_\e^k$ on $\Om_\e^k$.  In order to proof \eqref{supme}, it clearly suffices to show that for every $k\in \Z$ we have
\beq\label{supme2}
\sup_\e \int_{\Omega^k_\e} \left( | \nabla_{x'} m^k_\e|^2 + \frac{1}{\e^2} |\partial_{x_3} m^k_\e|^2 \right)\, dx<+\infty\,.
\eeq
To this aim, observe that for $k\geq 2$ we have
\begin{multline*}
\nabla m_\e^k(x', x_3)=\biggl(\nabla _{x'}m_\e^{k-1}\Bigl(x', 2 f_k\Bigl(\frac{x'}{\e}\Bigr)-x_3\Bigr)+\frac2\e\partial_{x_3}m_\e^{k-1}\Bigl(x', 2 f_k\Bigl(\frac{x'}{\e}\Bigr)-x_3\Bigr)\nabla f_k\Bigl(\frac{x'}{\e}\Bigr),\\ -\partial_{x_3}m_\e^{k-1}\Bigl(x', 2 f_k\Bigl(\frac{x'}{\e}\Bigr)-x_3\Bigr)\biggr)\,.
\end{multline*}
Thus, \eqref{supme2} follows easily by induction for $k\geq 2$ recalling that by  \eqref{supe} we have 
$$
\sup_\e \left\|\Bigl(\nabla_{x'} m^1_\e, \frac1\e \partial_{x_3} m_\e^1\Bigr)\right\|_{L^2(\Om_\e^1; \M^{3\times 3})}<+\infty\,.
$$
The proof for $k\leq 0$ is analogous.
\end{proof}
We are now ready to proof the $\Gamma$-liminf inequality for the exchange energy.  
 \begin{proposition}\label{prop:exch-liminf}
Let $m_0\in H^1(\omega; \mathbb{S}^2)$  and let $\{m_\e\}_\e$ be such that $m_\e\in H^1(\Om_\e; \mathbb{S}^2)$ for every $\e>0$ and 
\beq\label{conve}
\int_{\Om_\e}|m_\e(x)-m_0(x')|^2\, dx\to 0
\eeq
as $\e\to 0^+$. Then
\beq\label{exch-liminf}
\int_{\omega}g_{\mathrm{hom}}(\nabla_{x'} m_0)\, dx'\leq \liminf_{\e \to 0} \int_{\Omega_\e} \left( | \nabla_{x'} m_\e|^2 + \frac{1}{\e^2} 
|\partial_{x_3} m_\e|^2 \right)\, dx\,,
\eeq
where $g_{\mathrm{hom}}$ is the homogenized exchange energy density defined in  \eqref{ghom}.
\end{proposition}
 When $f_2=-f_1+a$ for some $a>0$, the above result is proven in \cite{braides00}. It is also clear that the methods of \cite{braides00} could be adapted to deal with thin films of the form \eqref{dom-gen}.
 However, for the reader's convenience we prefer to give here a simple self-contained proof based on the two-scale approach developed in \cite{neukamm10}.  Following \cite{neukamm10} (see also \cite{neukamm13}) we consider the following notion of two-scale convergence      adapted to the  3D-2D dimension reduction framework with the purpose of capturing the in-plane oscillations.
 
 \begin{definition}\label{def:2s}
 Let $\Om^M$ be as in Lemma~\ref{lm:extension}, let $H$ be a finite dimensional Hilbert space,  and let $\{g_\e\}\subset L^2(\Om^M; H)$ be $L^2$-bounded.  For any subsequence $\e_n\searrow 0$   we say that $\{g_{\e_n}\}$ two-scale converges to $g$, with 
 $g\in L^2\big(\Om^M; L^2(Q; H)\big)$, and we write $g_{\e_n}\2 g$, if 
 $$
 \lim_{n }\int_{\Om^M}\Bigl\langle g_{\e_n}(x), \psi\Bigl(x,\frac{x'}{\e_n}\Bigr)\Bigr\rangle\, dx=\int_{\Om^M}\int_{Q}\langle g(x, y'), \psi(x, y')\rangle\, dy'\, dx
 $$
 for all $\psi\in L^2(\Om^M; C_\#(Q; H))$. Here, $C_\#(Q; H)$ denotes the space of the $Q$-periodic continuous functions from $\R^2$ to $H$,   endowed with the sup norm on $Q$, and $\langle\cdot, \cdot\rangle$  stands for the scalar product of $H$. 
 \end{definition}
\begin{definition}\label{def:2stest}
Any function $\psi\in L^2(\Om^M; C_\#(Q; H))$ will be called an {\em admissible} test function for the two-scale convergence  defined in Definition~\ref{def:2s}. 
 \end{definition}
\begin{proof}[Proof of Proposition~\ref{prop:exch-liminf}]
Without loss of generality we may assume that \eqref{supe} holds.  Let $\{\tilde m_\e\}$ be the family of extensions provided by Lemma~\ref{lm:extension}. In particular \eqref{supme} holds and $\tilde m_\e\wto m_0$ weakly in $H^1(\Om^M; \S^2)$. Fix a subsequence $\e_n$ along which the liminf in \eqref{exch-liminf} is achieved. Thus,
 denoting by $\mathcal Y$ the subspace of $H^1((0,M)\times Q; \R^3)$ of functions $m=m(x_3, y' )$ that are $Q$-periodic in the $y'$-variable,  we may thus apply \cite[Theorem 6.3.3]{neukamm10}  and find $m_1=m_1(x', x_3, y')\in L^2(\omega; \mathcal Y)$ and a (not relabelled) subsequence such that 
 \beq\label{me2s}
 \Bigl(\nabla_{x'}\tilde m_{\e_n}, \frac1{\e_n}\pa_{x_3}\tilde m_{\e_n}\Bigr)\2 (\nabla_{x'}m_0+\nabla_{y'}m_1, \pa_{x_3}m_1 )
 \eeq
 in the sense of Definition~\ref{def:2s}, that is, 
 \begin{multline*}
\lim_n \int_{\Om^M}\Bigl\langle\Bigl(\nabla_{x'}\tilde m_{\e_n}(x), \frac1{\e_n}\pa_{x_3}\tilde m_{\e_n}(x)\Bigr),  \psi\Bigl(x,\frac{x'}{\e_n}\Bigr)\Bigr\rangle\, dx
\\=\int_{\Om^M}\int_{Q}\Bigl\langle\Big(\nabla_{x'}m_0(x')+\nabla_{y'}m_1(x, y'), \pa_{x_3}m_1(x,y') \Big),  \psi(x, y')\Bigr\rangle\, dy'\, dx
 \end{multline*}
 for all $\psi\in L^2(\Om^M; C_\#(Q; \M^{3\times3}))$.   For $\eta>0$   we can define
 $$
 m^\eta_1(x',x_3, y'):=\int_{\R^2}\rho_\eta(y'-  z')m_1(x',x_3, z')\, dz'  
 $$
for almost every $(x', x_3)\in \Omega^M$ and for all $ y'\in   Q$, where $(\rho_\eta)_\eta$ stands for the standard family of mollifiers on $\R^2$.  Note that in particular  
$\nabla_{x_3, y'}m_1^\eta\in L^2(\Om^M; C_\#(Q; \M^{3\times3}))$ for every $\eta>0$ and thus it can be used as a test function for the two-scale convergence, see Definition~\ref{def:2stest}.

 For every $k\in \N$, $x_3\in (0,M)$, and $y'\in \R^2$ set
 $$
 g_k(x_3, y'):=\inf\{\chi_{(f_1(z'), f_2(z'))}(t)+k|(t,z')-(x_3, y')|:\, t\in (0,M),\, z'\in \R^2\}\,,
 $$
 so that $0\leq g_k\leq 1$
 \beq\label{gkey}
 g_k(x_3, y')\nearrow g(x_3, y'):=\chi_{(f_1(y'), f_2(y'))}(x_3)
 \eeq
 as $k\to\infty$. Note also that by construction $g_k$ is $k$-Lipschitz continuous and $Q$-periodic in the $y'$-variable. Therefore, it is an admissible test function for the two-scale convergence.  Notice that for every $n$, $k\in \N$ and $\eta>0$ we have
 \begin{align*}
  \int_{\Omega_{\e_n}} \Bigl( | \nabla_{x'} m_{\e_n}|^2 & + \frac{1}{\e_n^2} |\partial_{x_3} m_{\e_n}|^2 \Bigr)\, dx \geq 
    \int_{\Omega^M}g_k\Bigr(x_3, \frac{x'}{\e_n}\Bigl) \left( | \nabla_{x'} \tilde m_{\e_n}|^2  + \frac{1}{\e_n^2} |\partial_{x_3} \tilde m_{\e_n}|^2 \right)\, dx\\
   &\geq - \int_{\Omega^M}g_k\Bigr(x_3, \frac{x'}{\e_n}\Bigl) \Bigl | \nabla_{x'}  m_{0}(x')+\nabla_{y'}m_1^\eta\Bigl(x,\frac{x'}{\e_n}\Bigr) \Bigr|^2\, dx\\
   &\quad +2\int_{\Omega^M}g_k\Bigr(x_3, \frac{x'}{\e_n}\Bigl) \Bigl \langle \nabla_{x'}  m_{0}(x')
   +\nabla_{y'}m_1^\eta\Bigl(x,\frac{x'}{\e_n}\Bigr), \nabla_{x'}\tilde m_{\e_n}(x)\Bigr\rangle\, dx\\
     &\quad- \int_{\Omega^M}g_k\Bigr(x_3, \frac{x'}{\e_n}\Bigl) \Bigl |\pa_{x_3}m_1^\eta\Bigl(x,\frac{x'}{\e_n}\Bigr) \Bigr|^2\, dx\\
      &\quad +2\int_{\Om^M}g_k\Bigr(x_3, \frac{x'}{\e_n}\Bigl)\Bigl \langle \frac1{\e_n}\pa_{x_3}\tilde m_{\e_n}(x) , \pa_{x_3}m_1^\eta\Bigl(x,\frac{x'}{\e_n}\Bigr) \Bigr\rangle\, dx
 \end{align*}
 Recalling  that
 $g_k\Bigl(\cdot, \frac{\cdot}{\e_n}\Bigr)\2 g_k$
 as $n\to \infty$, using \eqref{me2s} and the admissibility of  $\nabla_{x_3, y'}m_1^\eta$, $g_k$ as test functions for the two-scale convergence, we deduce that 
 \begin{align*}
  \liminf_n &\int_{\Omega_{\e_n}} \Bigl( | \nabla_{x'} m_{\e_n}|^2  + \frac{1}{\e_n^2} |\partial_{x_3} m_{\e_n}|^2 \Bigr)\, dx\\
  & \geq \int_{\Omega^M}\int_Q g_k(x_3, y') \biggl[ -| \nabla_{x'}  m_{0}(x')+\nabla_{y'}m_1^\eta(x,y') |^2\\
   &\quad\qquad\qquad\qquad\qquad +2\langle \nabla_{x'}  m_{0}(x')
   +\nabla_{y'}m_1^\eta(x,y'), \nabla_{x'}m_0(x')+\nabla_{y'}m_1(x, y') \rangle\, \\
     &\quad\qquad\qquad\qquad\qquad-  |\pa_{x_3}m_1^\eta(x,y') |^2
     + 2 \langle \pa_{x_3} m_{1}(x, y') , \pa_{x_3}m_1^\eta(x,y')\rangle\biggr]\,dy' dx\,.
\end{align*} 
In turn, recalling \eqref{gkey} and that $\nabla_{x_3, y'}m_1^\eta\to \nabla_{x_3, y'}m_1$ in $L^2(\Om^M; L^2(Q; \M^{3\times 3})$ as $\eta\to 0^+$, we may conclude 
\begin{align*}
  \liminf_n \int_{\Omega_{\e_n}} \Bigl( | \nabla_{x'} m_{\e_n}|^2  + \frac{1}{\e_n^2} |\partial_{x_3} m_{\e_n}|^2 \Bigr)\, dx & \geq 
  \int_{\Omega^M}\int_Q g(x_3, y') \Bigl[ | \nabla_{x'}  m_{0}(x')+\nabla_{y'}m_1(x,y') |^2\\
  &\qquad\qquad\qquad\qquad\qquad+  |\pa_{x_3}m_1(x,y') |^2\Bigr]\, dy'dx\\
  &= \int_\omega\int_{Q_{f_1,f_2}} \Bigl[ | \nabla_{x'}  m_{0}(x')+\nabla_{y'}m_1(x',x_3, y') |^2\\
  &\qquad\qquad\qquad+  |\pa_{x_3}m_1(x', x_3, ,y') |^2\Bigr]\, dy'dx_3 dx'\\
  &\geq \int_{\omega}g_{\mathrm{hom}}(\nabla_{x'} m_0)\, dx'\,,
  \end{align*}
  where the last inequality follows from the very definition \eqref{ghom} of $g_{\mathrm{hom}}$, recalling that for a.e. $x'\in \omega$ we have $m_1(x', \cdot, \cdot)\in
   H^1_{\#}(Q_{f_1,f_2};\,\R^3)$. This concludes the proof of the proposition.
\end{proof}
 
 We now  seek to prove the upper bound. We start with the following remark.
 
\begin{remark}[Cell formula revisited]\label{rm:cell} Let $\xi\in \M^{3\times 2}$ and denote by 
$\bar \xi$ its extension to a matrix in $\M^{3\times 3}$  obtained by taking a vanishing third column. Then by standard arguments    $\vphi_\xi$ solves the    minimization problem in  \eqref{ghom} if and only if it is a weak solution  to the problem
\beq\label{wmin}
\begin{cases}
\Delta \vphi_\xi=0 & \text{in $Q_{f_1, f_2}$,}\\
\nabla \vphi_\xi\nu=-\bar \xi\nu & \text{on $\pa Q_{f_1, f_2}\setminus (\pa Q\times \R)$,}\\
 \vphi_\xi \text{ is $Q$-periodic in the $y'$-variable,}
\end{cases}
\eeq
where $\nu$ denotes the outer unit normal to $\pa Q_{f_1, f_2}$. In particular, $\vphi_\xi$ is unique up adding constant vectors. Incidentally, \eqref{wmin} shows that the solution $\vphi_\xi$ is non constant, unless of course $f_1$ and $f_2$ are constant. 
Note also that whenever $\vphi_\xi$ is nonconstant then necesessarily it is also $x_3$-dependent, since otherwise it would be a periodic harmonic function (and thus constant).
Let now $s\in \S^2$ be such that $s\xi=0$ (that is, $s$ is orthogonal to both columns of $\xi$).  Then, setting $\psi_\xi:=\vphi_\xi- (\vphi_\xi\cdot s)s$ we can argue as in \cite[page 10]{alouges15}  to show that 
$$
|\xi+\nabla_{y'}\vphi_\xi|^2+|\pa_{y_3}\vphi_\xi|^2\geq |\xi+\nabla_{y'}\psi_\xi|^2+|\pa_{y_3}\psi_\xi|^2+|\nabla (\vphi_\xi\cdot s)|^2\geq  |\xi+\nabla_{y'}\psi_\xi|^2+|\pa_{y_3}\psi_\xi|^2\,.
$$
It follows that $\psi_\xi$ is also a solution and thus $\nabla (\vphi_\xi\cdot s)\equiv 0$, that is, $\vphi_\xi\cdot s$ is constant.
 Therefore,  upon  adding a suitable constant vector, we may assume that the solution $\vphi_\xi$ satisfies
$$
\begin{cases}
\vphi_\xi \text{ solves  \eqref{ghom} or equivalently \eqref{wmin}}\,,\\
\int_{Q_{f_1, f_2}}\vphi_\xi\, dx=0\,,\\
 \vphi_\xi\cdot s=0 \text{ in }Q_{f_1, f_2}\,.
\end{cases}
$$
The above conditions determine $\vphi_\xi$ uniquely. Notice also that $\vphi_{\lambda \xi}=\lambda\vphi_\xi$ for all $\lambda\in \R$, so that, in particular, $g_{\mathrm{hom}}(\lambda \xi)=|\lambda|^2 g_{\mathrm{hom}}(\xi)$ for all $\xi\in \M^{3\times 2}$ and for all $\lambda\in \R$.  Moreover, standard arguments show that if $\xi_k\to \xi$, then $\vphi_{\xi_k}\to \vphi_\xi$ in 
$H^1$ and thus  $g_{\mathrm{hom}}$ is continuous. Finally, choosing $\vphi=0$ as a test function in \eqref{ghom} we immediately get $g_{\mathrm{hom}}(\xi)\leq |\xi|^2$ for all $\xi\in \M^{3\times 2}$.

\end{remark}
 
\begin{lemma}\label{lm:density}
Let $M>0$ be as in Lemma~\ref{lm:extension} and denote by  $\mathcal Y$   the subspace of $H^1(Q\times (0,M); \R^3)$ of functions $m=m(y)$ that are $Q$-periodic in the $y'$-variable. Let $m_0\in C^1(\overline \omega; \S^2)$ then, for $g_{\mathrm{hom}}$ defined in \eqref{ghom}, the following identity holds:
\begin{multline}\label{switchinf}
\int_{\omega}g_{\mathrm{hom}}(\nabla_{x'} m_0)\, dx'=\inf\biggl\{ \int_\omega\int_{Q_{f_1,f_2}} \Bigl[ | \nabla_{x'}  m_{0}(x')+\nabla_{y'}m(x',y) |^2+  |\pa_{y_3}m(x', y) |^2\Bigr]\, dy  dx':\\
m\in C^1(\overline \omega;  \mathcal Y)\text{ s.t. } m(x', y)\cdot m_0(x')\equiv0\text{ for a.e.  }(x', y)\in \omega \times [Q\times (0,M)]  \biggr\}\,.
\end{multline}
\begin{proof}
Without loss of generality we may assume that $m_0\in C^1(\R^2; \S^2)$. Now for every $x'\in \R^2$ let $\overline m(x', \cdot)\in H^1_{\#}(Q_{f_1,f_2};\,\R^3)$ be the {\em unique}
solution to 
\beq\label{thirdeq}
\begin{cases}
\overline m(x', \cdot) \text{ solves } \eqref{ghom} \text{ with }\xi\text{ replaced by }\nabla_{x'} m_0(x')\,,\\
\int_{Q_{f_1, f_2}}\overline m(x', y) dy=0\,,\\
 \overline m(x', \cdot)\cdot m_0(x')=0 \text{ in }Q_{f_1, f_2}\,.
\end{cases}
\eeq
The solution to the above problem exists and is unique, thanks to Remark~\ref{rm:cell}, since $m_0$ is $\S^2$-valued and thus  $m_0(x')\nabla m_0(x')=0$ for all $x'$. By repeated reflections of $\overline m(x', \cdot)$ with respect to the $y_3$-variable (as in the proof of Lemma~\ref{lm:extension}) we may in fact assume that $\overline m(x',\cdot)\in \mathcal Y$ and that the third equation in \eqref{thirdeq} holds in $Q\times (0,M)$. 
Due to uniqueness, it is easy to see that $\overline m\in C^0(\R^2; \mathcal Y)$. In particular, $\overline m$ and $\nabla_y\overline m$ are globally measurable and 
\beq\label{meas}
\int_{\omega}g_{\mathrm{hom}}(\nabla_{x'} m_0)\, dx'= \int_\omega\int_{Q_{f_1,f_2}} \Bigl[ | \nabla_{x'}  m_{0}(x')+\nabla_{y'}\overline m(x',y) |^2+  |\pa_{y_3}\overline m(x', y) |^2\Bigr]\, dy  dx'\,.
\eeq
 Let $(\rho_\eta)_{\eta>0}$ be a family of standard mollifiers on $\R^2$ and for every $y\in Q\times (0,M)$ set  $\overline m_\eta(\cdot, y):=\rho_\eta*\overline m(\cdot, y)$, that is, $\overline m_\eta$ is defined by taking the convolution of $\overline m$ with respect to the $x'$-variable. Note that by the properties of convolutions we have $\overline m_\eta\in C^\infty(\R^2; \mathcal Y)$ and $\overline m_\eta\to \overline m$ in $C^0(\overline\omega; \mathcal Y)$, as $\eta\to 0^+$. In turn, setting 
 $\widehat m_\eta:=\overline m_\eta- (\overline m_\eta\cdot m_0)m_0$, we have
  $\widehat m_\eta\in C^1(\R^2;  \mathcal Y)$ and  $\widehat m_\eta(x', \cdot)\cdot m_0(x')\equiv0\text{ for all }x'$. Moreover, using the third equation in \eqref{thirdeq} in $Q\times (0,M)$ one sees that  
  $\widehat m_\eta\to \overline m - (\overline m \cdot m_0)m_0=\overline m$ in $C^0(\overline\omega;\mathcal Y)$ as $\eta\to 0^+$. Owing to the latter convergence property and recalling \eqref{meas} we easily deduce
  \begin{align*}
  &\int_{\omega}g_{\mathrm{hom}}(\nabla_{x'} m_0)\, dx'=\int_\omega\int_{Q_{f_1,f_2}} \Bigl[ | \nabla_{x'}  m_{0}(x')+\nabla_{y'}\overline m (x',y) |^2+  |\pa_{y_3}\overline m (x', y) |^2\Bigr]\, dy  dx'\\
  &= \lim_{\eta\to 0^+} \int_\omega\int_{Q_{f_1,f_2}} \Bigl[ | \nabla_{x'}  m_{0}(x')+\nabla_{y'}\widehat m_\eta(x',y) |^2+  |\pa_{y_3}\widehat m_\eta(x', y) |^2\Bigr]\, dy  dx'\\
  &\geq \inf\biggl\{ \int_\omega\int_{Q_{f_1,f_2}} \Bigl[ | \nabla_{x'}  m_{0}(x')+\nabla_{y'}m(x',y) |^2+  |\pa_{y_3}m(x', y) |^2\Bigr]\, dy  dx':
  \\
& \qquad\qquad\qquad\qquad\qquad\qquad m\in C^1(\overline \omega;  \mathcal Y )\text{ s.t. } m(x', \cdot)\cdot m_0(x')\equiv0\text{ for all }x'\in \omega \biggr\}\,.
    \end{align*}
Since the other inequality is trivial, this concludes the proof of the lemma. 
\end{proof}
\end{lemma}
 
 We are now ready to establish the upper bound for the limiting exchange energy.
 \begin{proposition}\label{prop:exch-limsup}
Let $m_0\in H^1(\omega; \mathbb{S}^2)$. Then there exists  $\{m_\e\}_{\e>0}$  such that $m_\e\in H^1(\Om_\e; \mathbb{S}^2)$ for every $\e>0$, \eqref{conve} holds and
$$
\limsup_{\e \to 0} \int_{\Omega_\e} \left( | \nabla_{x'} m_\e|^2 + \frac{1}{\e^2} 
|\partial_{x_3} m_\e|^2 \right)\, dx\leq \int_{\omega}g_{\mathrm{hom}}(\nabla_{x'} m_0)\, dx'\,.
$$
\end{proposition}
\begin{proof}
We start by assuming that $m_0\in C^1(\overline\omega; \mathbb{S}^2)$. Fix $\eta>0$. Then, by Lemma~\ref{lm:density} we may find 
$m\in C^1(\overline \omega;  \mathcal Y)$ such that 
\beq\label{st1000}
 m(x', \cdot)\cdot m_0(x')=0\text{ in $Q\times (0,M)$}\qquad\text{for all }x'\in \omega
 \eeq
 and 
 \beq\label{st1001}
\int_\omega\int_{Q_{f_1,f_2}} \Bigl[ | \nabla_{x'}  m_{0}(x')+\nabla_{y'}m(x',y) |^2+  |\pa_{y_3}m(x', y) |^2\Bigr]\, dy  dx'\leq \int_{\omega}g_{\mathrm{hom}}(\nabla_{x'} m_0)\, dx'+\eta\,.
 \eeq
 For every $\e>0$  and for $x\in \Om^M:=\omega\times (0,M)$ we set
 $$
 \widehat m_\e(x):=m_0(x')+\e m\Bigl(x', \frac{x'}\e, x_3\Bigr) \qquad\text{and}\qquad m_\e:=\frac{\widehat m_\e}{|\widehat m_\e|} .
 $$ 
 Clearly $\{m_\e\}\subset H^1(\Om^M; \S^2)$ and satisfies \eqref{conve}. 
Since by \eqref{st1000} we have $|\widehat m_\e|\geq 1$, it is immediately  checked that 
 $$
  | \nabla_{x'} m_\e|^2 + \frac{1}{\e^2} 
|\partial_{x_3} m_\e|^2\leq  | \nabla_{x'}\widehat  m_\e|^2 + \frac{1}{\e^2} 
|\partial_{x_3} \widehat m_\e|^2\,.
 $$
 Thus,   setting $g(x_3, y'):=\chi_{(f_1(y'), f_2(y'))}(x_3)$, we may estimate 
 \begin{align*}
 \limsup_{\e \to 0} & \int_{\Omega_\e} \left( | \nabla_{x'} m_\e|^2 + \frac{1}{\e^2} 
|\partial_{x_3} m_\e|^2 \right)\, dx \\
&\leq \lim_{\e \to 0}\int_{\Om^M}g\Bigl(x_3, \frac{x'}{\e}\Bigr)\left( | \nabla_{x'}\widehat  m_\e|^2 + \frac{1}{\e^2} 
|\partial_{x_3} \widehat m_\e|^2\right)\, dx\\
&=  \lim_{\e \to 0}\int_{\Om^M}g\Bigl(x_3, \frac{x'}{\e}\Bigr)\left(\Bigl | \nabla_{x'}m_0(x')+ \nabla_{y'}m\Bigl(x', \frac{x'}\e, x_3\Bigr)\Bigr|^2 +  
\Bigr|\partial_{y_3} m\Bigl(x', \frac{x'}\e, x_3\Bigr)\Bigr|^2\right)\, dx\\
&= \int_{\Om^M}\int_{Q}g(x_3, y)\left(  | \nabla_{x'}m_0(x')+ \nabla_{y'}m(x', y', x_3)|^2 +  
|\partial_{y_3} m(x', y', x_3)|^2\right)\,dy' dx\\ 
&\leq  \int_{\omega}g_{\mathrm{hom}}(\nabla_{x'} m_0)\, dx'+\eta
 \end{align*}
 where the last equality has been obtained by passing to the two-scale limit, while the last inequality is   \eqref{st1001}.
  By the arbitrariness of $\eta$ and a standard diagonalization argument the thesis of the proposition is established when  $m_0\in C^1(\overline\omega; \mathbb{S}^2)$.
 
Let now   $m_0\in H^1( \omega; \mathbb{S}^2)$. Then there exists $\{m_k\}\subset C^1(\overline\omega; \mathbb{S}^2)$ such that $m_k\to m_0$ in $H^1( \omega; \mathbb{S}^2)$ as $k\to \infty$. In particular, recalling that  $g_{\mathrm{hom}}$ is continuous and  $g_{\mathrm{hom}}(\xi)\leq |\xi|^2$ (see Remark~\ref{rm:cell}), we have 
$$
 \int_{\omega}g_{\mathrm{hom}}(\nabla_{x'} m_k)\, dx'\to  \int_{\omega}g_{\mathrm{hom}}(\nabla_{x'} m_0)\, dx'\,.
$$
The thesis follows by applying the first part of the proof to each $m_k$ and using  diagonalization argument.  
\end{proof}

\subsection{$\Gamma$-convergence}\label{sbs:G}
In this section we  prove the main compactness and  $\Gamma$-convergence theorem by combining   all the previous results.
\begin{proof}[Proof of Theorem\onda\ref{th:main}]
We start by showing part (i).  Let $\{m_\e\}$ be as in the statement and for every $\e>0$ let  $\overline M_\e$ be the function in $H^1(\omega; \R^3)$, with $|\overline M_\e|\leq 1$ defined by
$$
\overline M_\e(x'):=\medint_{f_1(x'/\e)}^{f_2(x'/\e) }m_\e(x', x_3)\, dx_3=\medint_{\e f_1(x/\e')}^{\e f_2(x'/\e) }M_\e(x', x_3)\, dx_3\,,
$$
where, we recall, $M_\e(x', x_3):=m_\e(x', x_3/\e)$. Note that, in particular,   \eqref{supe} holds.
Using \eqref{supe} it is straightforward to check that $\{\overline M_\e\}$ is bounded in $H^1(\omega; \R^3)$.
Thus, up to a (not relabelled) subsequence there exists $m_0\in H^1(\omega; \R^3)$ such that 
$\overline M_\e\wto m_0$ weakly in $H^1(\omega; \R^3)$. Observe now that by the one-dimensional Poincar\'e-Wirtinger's inequality we have 
\begin{multline*}
\int_{\Om_\e}|m_\e-\overline M_\e|^2\, dx=\int_{\omega}\int_{f_1(x'/\e)}^{f_2(x'/\e)}\biggl|m_\e(x', x_3)-\medint_{f_1(x'/\e)}^{f_2(x'/\e)}m_\e(x', x_3)\, dx_3\biggr|^2\, dx'\\
\leq \int_\omega\frac{(f_2(x'/\e)-f_1(x'/\e))^2}{\pi^2} \int_{f_1(x'/\e)}^{f_2(x'/\e)}|\pa_{x_3} m_\e|^2\, dx_3dx'\leq C\e^2\,,
\end{multline*}
thanks to \eqref{supe}, for some constant $C>0$ independent of $\e$. We deduce that
$$
\int_{\Om_\e}|m_\e-m_0|^2\, dx\to 0\,.
$$

For part (ii),  we may assume without loss of generality
that 
$$
\liminf_\e E_\e(m_\e)=\lim_\e E_\e(m_\e)<+\infty\,,
$$
In particular, \eqref{supe} holds.  
Thus, defining $\overline M_\e$ as before, we deduce that $\{\overline M_\e\}$ is bounded in $H^1$. By \eqref{conve1}
  it readily follows that 
\beq\label{h1b2}
\overline M_\e\wto m_0\qquad\text{weakly in $H^1(\omega; \R^3)$.}
\eeq
In turn, by Lemma~\ref{lm:media} (and Remark~\ref{rm:media}) and Proposition~\ref{prp:finalgeneral} we get
\beq\label{get100}
\frac1\e\int_{\R^3}|\nabla u_\e|^2\, dx\to  \int_\omega  A_{\mathrm{hom}}\, m_0\cdot m_0\, dx'\,,
\eeq
which together with Proposition~\ref{prop:exch-liminf} yields the conclusion of part (ii). 

Part (iii) easily follows from Proposition~\ref{prop:exch-limsup} and the fact that \eqref{get100} holds whenever \eqref{conve1} and \eqref{supe} hold. 
\end{proof}
\begin{proof}[Proof of Corollary\onda\ref{cor:main}]
By Theorem~\ref{th:main}  and standard $\Gamma$-convergence arguments we infer that there exists a global minimizer $m_0$ of $E_0$
in $H^1(\omega; \mathbb{S}^2)$ such that, up to a (not relabelled) subsequence, \eqref{conve1} holds.  It is now easy to see that  $m_0$ is a global minimizer if and only if it is constant and minimizes the quadratic form associated to the matrix $A_{\mathrm{hom}}$. This concludes the proof of the corollary.
\end{proof}

\paragraph*{\bf Acknowledgements}  The authors thank O. Tretiakov for helpful discussions and acknowledge the support from EPSRC grant EP/K02390X/1 and Leverhulme grant RPG-2014-226.

\bibliography{anisotropy_curvature,nonlin}
\bibliographystyle{plain}

\end{document}